\documentclass[a4paper,12pt]{article}
\usepackage[usenames,dvipsnames]{xcolor}
\usepackage[utf8]{inputenc}
\usepackage[english]{babel}
\usepackage{amsmath}
\usepackage{amssymb}
\usepackage{amsthm}
\usepackage{mathrsfs}
\usepackage{dsfont}
\usepackage{soul}
\usepackage[mathcal]{euscript}
\usepackage{stmaryrd}
\usepackage{multirow}
\usepackage{authblk}
\usepackage[top=3cm, bottom=3cm, left=3cm, right=3cm]{geometry}

\usepackage{tikz}
\usetikzlibrary{arrows,automata}

\usepackage{pbox}
\usepackage[colorlinks=true,urlcolor=blue,linkcolor=blue]{hyperref}
\usepackage[round]{natbib}

\usepackage{caption}

\SetSymbolFont{stmry}{bold}{U}{stmry}{m}{n}
\def\mathscr{\EuScript}

\newcommand{\sect}[1]{Sect.~\ref{#1}}

\newcommand{\aggreg}{A}
\newcommand{\factor}{F}
\newcommand{\composite}{S}
\newcommand{\subaggreg}{\composite^{\aggreg,\factor}}

\newcommand{\HeadSet}{\mathbb{H}}
\newcommand{\TailSet}{\mathbb{T}}
\newcommand{\headelement}{h}
\newcommand{\tailelement}{t}

\newcommand{\valueset}{\mathbb{A}}
\newcommand{\secondvalueset}{\FF}
\newcommand{\valueelement}{a}
\newcommand{\secondvalueelement}{f}

\newcommand{\cA}{\mathcal{A}}

\newcommand{\cP}{\mathcal{P}}

\newcommand{\cX}{\mathcal{X}}

\newcommand{\BB}{\mathbb{B}}

\newcommand{\FF}{\mathbb{F}}

\newcommand{\NN}{\mathbb{N}}

\newcommand{\RR}{\mathbb{R}}

\newcommand{\XX}{\mathbb{X}}
\newcommand{\YY}{\mathbb{Y}}
\newcommand{\ZZ}{\mathbb{Z}}

\newcommand{\np}[1]{(#1)}                                   % Parenth\`{e}se normal
\newcommand{\bp}[1]{\big(#1\big)}                           % Parenth\`{e}se big
\newcommand{\Bp}[1]{\Big(#1\Big)}                           % Parenth\`{e}se Big
\newcommand{\nc}[1]{[#1]}                                   % Crochet normal
\newcommand{\bc}[1]{\big[#1\big]}                           % Crochet big
\newcommand{\Bc}[1]{\Big[#1\Big]}                           % Crochet Big
\newcommand{\na}[1]{\{#1\}}                                 % Accolade normal
\newcommand{\ba}[1]{\big\{#1\big\}}                         % Accolade big
\newcommand{\Ba}[1]{\Big\{#1\Big\}}                         % Accolade Big
\newcommand{\bga}[1]{\bigg\{#1\bigg\}}                      % Accolade bigg
\newcommand{\Bga}[1]{\Bigg\{#1\Bigg\}}                      % Accolade Bigg

\newcommand{\tribu}[1]{\mathscr{#1}}                        % Tribu
\newcommand{\omeg}{\Omega}                                  % espace du triplet
\newcommand{\prbt}{\mathbb{P}}                              % proba  du triplet
\newcommand{\espe}{\mathbb{E}}                              % Symbole esp\'{e}rance
\newcommand{\mesr}{\mathbf{\rho}}                              % Symbole mesure de risque
\newcommand{\mesu}{\mathbf{\rho}}                              % Symbole mesure de risque

\makeatletter
\def\va@a{\boldsymbol{\va@arg^{\textstyle\text{\unboldmath$\scriptstyle\va@expo$}}_{\textstyle\text{\unboldmath$\scriptstyle\va@index$}}}}
\def\va#1{\def\va@expo{}\def\va@index{}\def\va@arg{\uppercase{#1}}%
  \@ifnextchar^{\va@h}{\@ifnextchar_\va@u\va@a}}
\def\va@h^#1{\def\va@expo{#1}\@ifnextchar_\va@hu\va@a}
\def\va@u_#1{\def\va@index{#1}\@ifnextchar^\va@uh\va@a}
\def\va@hu_#1{\def\va@index{#1}\va@a}
\def\va@uh^#1{\def\va@expo{#1}\va@a}
\makeatother

\newcommand{\normdelim}[1]{\nc{#1}}                         % Taille ``normal''
\newcommand{\bigdelim}[1]{\bc{#1}}                          % Taille ``big''
\newcommand{\vardelim}[1]{\left(#1\right)}                  % Taille ``variable''

\newcommand{\normdelims}[2]{\normdelim{#1\mid#2}}           % avec s\'{e}parateur
\newcommand{\bigdelims}[2]{\bigdelim{#1\ \big|\ #2}}        %
\newcommand{\nesp}[2]{\espe_{#1}\normdelim{#2}}           % Esp\'{e}rance normal
\newcommand{\besp}[2]{\espe_{#1}\bigdelim{#2}}            % Esp\'{e}rance big
\newcommand{\nespc}[3]{\espe_{#1}\normdelims{#2}{#3}}     % Esp\'{e}rance cond. normal
\newcommand{\bespc}[3]{\espe_{#1}\bigdelims{#2}{#3}}      % Esp\'{e}rance cond. big
\def\eqsepv{\; , \enspace}                                  % Virgule dans une \'{e}quation
\def\eqfinv{\; ,}                                           % Virgule en fin d'\'{e}quation
\def\eqfinp{\; .}                                           % Point en fin d'\'{e}quation
\newcommand{\finpreuvesymb}{$\Box$}%symbole fin preuve
\newcommand{\finremarksymb}{$\Diamond$}%symbole fin remarque
\newcommand{\finexemplesymb}{$\triangle$}%symbole fin exemple
\newcommand{\finpreuve}{\ \hspace*{\fill}\finpreuvesymb}
\newcommand{\finremark}{\ \hspace*{\fill}\finremarksymb}
\newcommand{\finexemple}{\ \hspace*{\fill}\finexemplesymb}

\newcommand{\image}[1]{\textrm{Im}\np{#1}}

\newcommand{\PRIMAL}{{\mathbb C}}
\newcommand{\primal}{c}
\newcommand{\DUAL}{{\PRIMAL}^{\sharp}}
\newcommand{\dual}{{\primal}^{\sharp}}
\newcommand{\barRR}{[-\infty,+\infty]} %{\RR \cup \{+\infty\}}
\newcommand{\fonctionprimal}{f} %{\mathbb F}
\newcommand{\coupling}{\Phi}
\newcommand{\SFM}[2]{#1^{#2}}
\newcommand{\LowPlus}{\plusdot}
\newcommand{\UppPlus}{\dotplus}  

\def\stackops#1#2#3{%
  \mathrel{\vbox{\offinterlineskip\ialign{%
        \hfil##\hfil\cr
        $#1$\cr
        \noalign{\kern#3}
        $#2$\cr}}}}

\def\plusdot{\stackops{\cdot}{+}{-2.5ex}}
\makeatletter
\def\endproof{\finpreuve\@endtheorem}
\def\endremark{\finremark\@endtheorem}
\def\endexample{\finexemple\@endtheorem}
\makeatother

\title{Equivalence Between Time Consistency\\and Nested Formula}
\author[ ]{Henri Gérard\textsubscript{a,b}\thanks{hgerard.pro@gmail.com}}
\author[ ]{Michel De Lara\textsubscript{a}\thanks{michel.delara@enpc.fr}}
\author[ ]{Jean-Philippe Chancelier\textsubscript{a}\thanks{jean-philippe.chancelier@enpc.fr}}
\affil[ ]{\textsubscript{a}\small \it Université Paris-Est, CERMICS (ENPC), F-77455 Marne-la-Vallée, France}
\affil[ ]{\textsubscript{b}\small \it Université Paris-Est, Labex Bézout, F-77455 Marne-la-Vallée, France}

\begin{document}
\maketitle
\begin{abstract}
Figure out a situation where, 
at the beginning of every week, one has to rank every pair of 
stochastic processes starting from that week up to the horizon.
Suppose that two processes are equal at the beginning of the week.
The ranking procedure is time consistent if the ranking
does not change between this week and the next one.
In this paper, we propose a minimalist definition of Time Consistency (TC)
between two (assessment) mappings.
With very few assumptions, we are able to prove an 
equivalence between Time Consistency and a Nested Formula (NF)
between the two mappings.
Thus, in a sense, two assessments are consistent if and only if one 
is factored into the other. 
We review the literature and observe that the various 
definitions of TC (or of NF) are special cases of ours, as they always include 
additional assumptions. 
By stripping off these additional assumptions, we present an overview of the
literature where the specific contributions of authors are enlightened.
Moreover, we present two classes of mappings, 
translation invariant mappings and Fenchel-Moreau conjugates, 
that display time consistency under suitable assumptions. 
\end{abstract}


\theoremstyle{definition}
\newtheorem{mydef}{Definition}[section]
\newtheorem{ex}[mydef]{Example}

\theoremstyle{plain}
\newtheorem{theo}[mydef]{Theorem}
\newtheorem{lem}[mydef]{Lemma}
\newtheorem{cor}[mydef]{Corollary}
\newtheorem{prop}[mydef]{Proposition}

\theoremstyle{remark}
\newtheorem{rem}[mydef]{Remark}
\newtheorem{ax}[mydef]{Axiom}
\newtheorem{assu}[mydef]{Assumption}

\newpage
\section{Introduction}
\label{Introduction}

Behind the words ``Time Consistency'' and ``Nested Formula'', 
one can find a vast literature resorting to economics, 
dynamical risk measures and stochastic optimization.

Let us start with economics.
In a dynamic bargaining problem, a group of agents has to agree on a common path
of actions. As time goes on and information is progressively revealed, 
they can all reconsider the past agreement, 
and possibly make new assessments leading to new actions. 
Stability is the property that the agents will stick to their previous
commitment.
Time consistency is a form of stability when an individual makes a deal between 
his different selves (agents) along time. 
The notion of ``consistent course of action''~\citep[see][]{Peleg-Yaari:1973}
is well-known in the field of economics, with the seminal work
of~\cite*{Strotz_RoES_1956}: an individual having planned
his consumption trajectory is consistent if, reevaluating his plans later on,
he does not deviate from the originally chosen plan. This idea of consistency
as ``sticking to one's plan'' may be extended to the uncertain case where
plans are replaced by decision rules (``Do thus-and-thus if you find
yourself in this portion of state space with this amount of
time left'', Richard Bellman cited in~\cite*{Dreyfus:2002});
\cite*{Hammond_RoES_1976} addresses ``consistency'' and
``coherent dynamic choice'', \cite*{Kreps-Porteus:1978}
refers to ``temporal consistency''. 
Another classical reference in economics is \cite*{Epstein-Schneider:2003}.

Dynamic or Time Consistency has been introduced in the context 
of dynamical risk measures \citep[see][for definitions
and properties of coherent and consistent dynamic risk
measures]{Riedel:2004,Detlefsen_FS_2005,Cheridito_EJP_2006,ADEH-coherent:2007}.

In the field of stochastic optimization, 
Time Consistency has then been studied 
for Markov Decision Processes by \cite*{ruszczynski2010risk}. 
\bigskip

These different origins of Time Consistency contribute to a disparate
literature.
First, as Nested Formulas lead naturally to Time Consistency,
some authors study the conditions to obtain Nested Formulas,
whereas others focus on the axiomatics of Time Consistency 
and obtain Nested Formulas.
Second, many definitions cohabit. For instance, 
\cite*{ruszczynski2010risk} add 
translation invariant property with additive criterion, 
\cite*{Shapiro:2016,ADEH-coherent:2007} 
add assumptions of coherent risk measures, 
and many authors focus on a particular structure of information (filtration).
In this disconnected landscape,  \cite*{DeLara-Leclere:2016}
tries to make the connection 
between ``dynamic consistency'' for optimal control problems 
(economics, stochastic optimization)
and ``time consistency'' for dynamic risk measures.
In this paper, we will focus on Time Consistency, motivated by 
dynamic risk measures --- where the future 
assessment of a tail of a process is consistent with the initial assessment 
of the whole process, head and tail --- but not limited to them.
Below, we sketch our definitions of TC and NF.
Our main contribution will be proving their equivalence.
Let $\HeadSet$ and $\TailSet$ be two sets, respectively called
\emph{head set} and \emph{tail set}.
Let $\valueset$, $\secondvalueset$ be two sets
and let $\aggreg$ and $\factor$ be two mappings as follows:
\begin{equation}
    \aggreg: \HeadSet \times \TailSet \to \valueset 
    \eqsepv
    \factor:\TailSet \to \secondvalueset     
    \eqfinp
\end{equation}
The mapping~$\aggreg$ is called an \emph{aggregator},
as it aggregates head-tail in $\HeadSet \times \TailSet$ into an element of~$\valueset$. 
The mapping~$\factor$ is called a \emph{factor} because of the 
Nested Formula~(NF).

\paragraph{Axiomatic for Time Consistency.}
We start presenting axiomatic of Time Consistency in a nutshell.
Depending on the authors, the objects that are manipulated are either processes 
\citep*{Riedel:2004,Detlefsen_FS_2005,Cheridito_EJP_2006,ADEH-coherent:2007}
or lotteries \citep*{Kreps-Porteus:1978,Epstein-Schneider:2003}.
These objects are divided into two parts: a head~$\headelement$ and 
a tail~$\tailelement$. 
On the one hand, we have a way to assess any tail~$\tailelement$
by means of a mapping~$\factor$ (factor), 
yielding~$\factor\np{\tailelement}$.  
On the other hand, we have a way to assess any couple
head-tail~$\np{\headelement, \tailelement}$ by means of a 
mapping~$\aggreg$ (aggregator), 
yielding~$\aggreg\np{\headelement, \tailelement}$.

We look for a consistency property between these two ranking mappings~$\factor$ 
and~$\aggreg$: if a tail~$\tailelement$ is equivalent to
a tail~$\tailelement'$, then the two elements 
$\np{\headelement,     \tailelement}$ 
and~$\np{\headelement, \tailelement'}$ --- that share the same head --- 
must be such that $\np{\headelement, \tailelement}$ is 
equivalent to 
$\np{\headelement, \tailelement'}$. 
This can be written mathematically as
\begin{equation*}
    \factor\np{\tailelement} = \factor\np{\tailelement'}
    \Rightarrow
\aggreg\np{\headelement,\tailelement} = \aggreg\np{\headelement,\tailelement'}
    \eqsepv
    \forall \np{\headelement,\tailelement,\tailelement'} 
    \in \HeadSet \times \TailSet^{2}
    \eqfinp \tag*{(TC)}
\end{equation*}

\paragraph{Axiomatic for Nested Formulas.}
Some authors focus on sufficient conditions to obtain a Nested Formula
\citep*{Shapiro:2016,Ruszczynski-Shapiro:2006}.
In a Nested Formula, the assessment~$\factor\np{\tailelement}$
of any tail~$\tailelement$ is factored
inside the assessment~$\aggreg\np{\headelement, \tailelement}$ of any 
head-tail~$\np{\headelement, \tailelement}$ by means 
of a surrogate mapping~$\subaggreg$ as follows:
\begin{equation*}
    \aggreg\np{\headelement, \tailelement}
    =
    \subaggreg\bp{\headelement, \factor\np{\tailelement}}
    \eqfinp \tag*{(NF)}
\end{equation*}

Of course, (NF) implies (TC).
We will prove the reverse: (TC) implies that there 
exists a mapping $\subaggreg$ such that (NF) holds true.
\bigskip

In Sect.~\ref{section_literature_review}, 
we go through the literature, with the goal of extracting 
the following components: what kind of objects are treated, 
what are the heads and the tails, how these objects are ranked.
In Sect.~\ref{section_results}, we formally state our definitions
of Time Consistency (TC) and Nested Formula (NF), 
and we prove their equivalence. 
We also provide conditions to obtain analytical properties of the mapping $\subaggreg$ appearing
in the Nested Formula, such as monotonicity, 
continuity, convexity, positive homogeneity and translation invariance.
In Sect.~\ref{section_application}, we show that our framework 
covers the different frameworks reviewed 
in Sect.~\ref{section_literature_review}.
Finally, in Sect.~\ref{Two_classes_of_time_consistent_mappings},
we present two classes of mappings, 
translation invariant mappings and Fenchel-Moreau conjugates, 
that display time consistency under suitable assumptions.

\section{Review of the literature}
\label{section_literature_review}

We have screened a selection of papers, in mathematics and economics, touching
Time Consistency and Nested Formula in various settings. 
Depending on the setting, we identify the following components, as introduced 
in Sect.~\ref{Introduction}: what kind of objects are treated, 
what are the heads and the tails, how are these objects ranked.
Table~\ref{table_TC_NF} sums up our survey. 

\begin{table}[h]
   \centering
   \begin{tabular}{c|c|c|c|c|c|}
       \cline{2-6}
        
        & Article
        & Objects 
        & Head 
        & Tail 
        & Assessment
        \\
        \hline
        \hline
        \multirow{4}*{\rotatebox[origin=c]{90}{$\overbrace{\qquad \qquad \qquad \quad}^{\text{\small{Time Consistency}}}$}}
        & \pbox{2.5cm}{\cite*{Kreps-Porteus:1978}}
        & Lottery
        & \pbox{4cm}{Lottery \\ from $1$ to $s$}
        & \pbox{4cm}{Lottery from\\$s+1$ to $T$}
        & \small{Expected utility}
        \\
        \cline{2-6}
        
        & \pbox{2.5cm}{\cite*{Epstein-Schneider:2003}}
        & Lottery
        & \pbox{4cm}{Lottery \\ from $1$ to $s$}
        & \pbox{4cm}{Lottery from\\$s+1$ to $T$}
        & \pbox{4cm}{\small{Not necessarily}\\\small{expected utility}}
        \\
        \cline{2-6}
        
        & \pbox{2.5cm}{\cite*{ruszczynski2010risk}}
        & Process
        & \pbox{4cm}{Process\\from $1$ to $s$}
        & \pbox{4cm}{Process from\\$s+1$ to $T$}
        & \pbox{4cm}{Dynamic\\risk measure}
        \\
        \cline{2-6}
        
        &\pbox{2.5cm}{\cite*{ADEH-coherent:2007}}
        & Process
        & \pbox{4cm}{Process\\from $1$ to $\tau$}
        & \pbox{4cm}{\small{Process from}\\\small{$\tau$ to $T$,}\\\footnotesize{$\tau$ stopping time}}
        & \pbox{4cm}{Coherent\\risk measure}
        \\
        \hline
        \hline
        \multirow{3}*{\rotatebox[origin=c]{90}{$\overbrace{\qquad \qquad \quad}^{\overset{\text{\small Nested}}{\text{\small Formula}}}$}}
        & \pbox{2.5cm}{\cite*{Shapiro:2016}}
        & Process
        & \pbox{4cm}{Process\\from $1$ to $s$}
        & \pbox{4cm}{Process from\\ $s+1$ to $T$}
        & \pbox{4cm}{Coherent\\risk measure}
        \\
        \cline{2-6}
        
        & \pbox{2.5cm}{\cite*{Ruszczynski-Shapiro:2006}}
        & Process
        & \pbox{4cm}{Process\\from $1$ to $s$}
        & \pbox{4cm}{Process from\\ $s+1$ to $T$}
        & \pbox{4cm}{Coherent\\risk measure}
        \\
        \cline{2-6}
        
        & \pbox{2.5cm}{\cite*{DeLara-Leclere:2016}}
        & Process
        & \pbox{4cm}{Process\\from $1$ to $s$}
        & \pbox{4cm}{Process from\\ $s+1$ to $T$}
        & \pbox{4cm}{Dynamic\\risk measure}
        \\
        \cline{2-6}
   \end{tabular}
\caption{Sketch of papers selected on Time Consistency and 
Nested Formulas 
\label{table_TC_NF}}
\end{table}

\subsection{Axiomatic for Time Consistency (TC)}

The first group of authors is subdivided between economists, 
who deal with lotteries and preferences, 
and probabilists who deal with stochastic processes
and dynamical risk measures. 
 
\subsubsection{Lotteries and preferences}

In \cite*{Kreps-Porteus:1978}, \cite*{Kreps-Porteus:1979} 
and \cite*{Epstein-Schneider:2003}, the authors deal with 
lotteries and preferences. 
A preference is a total, transitive and reflexive relation.
Proper assumptions make it possible that the preference relation
can be represented by a numerical evaluation. 
Assumptions of monotonicity and convexity are also made.

In~\cite*{Kreps-Porteus:1978}, the authors
propose axioms that make that 
the preference is represented by an expected utility formula.

By contrast, more general numerical representations are studied 
in~\cite*{Epstein-Schneider:2003}, even if the authors
add an hypothesis of additive criterion.
A summary of the assumptions can be found in Table~\ref{table_assumptions}.

\subsubsection{Dynamic risk measures and processes}

In~\cite*{ruszczynski2010risk} and \cite*{ADEH-coherent:2007}, 
the authors deal with 
stochastic processes assessed by dynamical risk measures.

In~\cite*{ruszczynski2010risk}, the author
studies a family of conditional risk measures which are monotonic, 
invariant by translation and homogeneous. The criterion is additive.

In~\cite*{ADEH-coherent:2007}, 
the authors
focus on the value of the stochastic process at the final time step. 
They use as assessment a particular class of risk measures,
the so-called coherent risk measures.

\subsection{Axiomatic for Nested Formulas (NF)}

In~\cite*{Shapiro:2016}, \cite*{Ruszczynski-Shapiro:2006} and 
\cite*{DeLara-Leclere:2016}, the focus is on exhibiting
sufficient conditions to obtain Nested Formulas.
All authors study stochastic processes, with an assumption of monotonicity
for the assessment, but there are some differences. 

In~\cite*{Ruszczynski-Shapiro:2006}, the authors
study coherent risk measures in their dual form (hence with properties
of convexity, invariance by translation and additive criterion).

In~\cite*{Shapiro:2016}, the author
focuses on assessing the value of the process at the final step
with coherent risk measures.

In~\cite*{DeLara-Leclere:2016}, the author
study how commutation properties between 
time aggregators and uncertainty aggregators make it possible to obtain 
Nested Formulas.

\begin{table}[h]
   \centering
   \begin{tabular}{c|c|c|c|c|c|c|c|c|}
       \cline{2-5}
        
        & Article
        & Monotonicity
        & \pbox{3cm}{Translation\\invariance} 
        & \pbox{3cm}{Convexity}
        \\
        \hline
        \hline
        \multirow{4}*{\rotatebox[origin=c]{90}{$\overbrace{\qquad \qquad \qquad \quad}^{\overset{\text{\small Time}}{\text{\small Consistency}}}$}}
        &\cite*{Kreps-Porteus:1978}
        & Yes
        & No
        & Yes
        \\
        \cline{2-5}
        
        &\cite*{Kreps-Porteus:1979} 
        & Yes
        & No
        & Yes
        \\
        \cline{2-5}
        
        &\cite*{Epstein-Schneider:2003}
        & Yes
        & No
        & Yes
        \\
        \cline{2-5}
        
        &\cite*{ruszczynski2010risk}
        & Yes
        & Yes
        & No
        \\
        \cline{2-5}
       
        &\pbox{5cm}{\cite*{ADEH-coherent:2007}}
        & Yes
        & Yes
        & Yes
        \\
        \hline
        \hline
        \multirow{3}*{\rotatebox[origin=c]{90}{$\overbrace{\qquad \quad}^{\overset{\text{\small Nested}}{\text{\small Formula}}}$}}
        &\cite*{Shapiro:2016}
        & Yes
        & Yes
        & Yes
        \\
        \cline{2-5}
        
        &\cite*{Ruszczynski-Shapiro:2006}
        & Yes
        & Yes
        & Yes
        \\
        \cline{2-5}
        
        &\cite*{DeLara-Leclere:2016}
        & Yes
        & No
        & No
        \\
        \cline{2-5}
   \end{tabular}
   \caption{Most common assumptions in the selection of papers on Time Consistency and Nested Formula \label{table_assumptions}}
\end{table}

\section{Main result: equivalence between time consistency and nested formula}
\label{section_results}

In Sect.~\ref{Introduction}, we have sketched the notions of 
Time Consistency and Nested Formula.
Now, in~\S\ref{subsection_WTC}, we properly define Weak Time Consistency --- 
with minimal assumptions --- and we prove 
that it is equivalent to a Nested Formula.
In~\S\ref{subsection_USTC}, we extend definitions and results to 
Usual and Strong Time Consistency: by adding order structures, we obtain additional properties.
In~\S\ref{subsection_properties}, we provide conditions to obtain analytical properties of the mapping appearing
in the Nested Formula, such as monotonicity, 
continuity, convexity, positive homogeneity and translation invariance.
Let us introduce basic notations. 
\bigskip

Let $\HeadSet$ and $\TailSet$ be two sets, respectively called
\emph{head set} and \emph{tail set}.
Let $\valueset$, $\secondvalueset$ be two sets
and let $\aggreg$ and $\factor$ be two mappings as follows:
\begin{equation}
    \aggreg: \HeadSet \times \TailSet \to \valueset \eqsepv
    \factor:\TailSet \to \secondvalueset     \eqfinp
    \label{eq:factor_aggregator}
\end{equation}
The mapping~$\aggreg$ is called an \emph{aggregator},
as it aggregates head-tail in $\HeadSet \times \TailSet$ into an element of~$\valueset$. 
The mapping~$\factor$ is called a \emph{factor} because of the 
Nested Formula~(NF) in Sect.~\ref{Introduction}. 

\begin{mydef}
With the couple aggregator-factor
$\np{\aggreg, \factor}$ in~\eqref{eq:factor_aggregator}
we associate the set-valued mapping
\begin{equation}
\label{definition_subaggregator}
    \begin{array}{rl}
        \subaggreg:\ 
        \HeadSet \times \image{\factor} 
        \rightrightarrows& 
        \valueset
        \\
        \np{\headelement, \secondvalueelement} 
        \mapsto
        & 
        \subaggreg\np{\headelement,\secondvalueelement} = \ba{\aggreg\np{\headelement,\tailelement} \mid
        \tailelement \in \factor^{-1}\np{\secondvalueelement}} \eqfinv
    \end{array}
\end{equation}
where $\image{\factor}=\factor\np{\TailSet}$.
We call $\subaggreg$ the \emph{subaggregator} 
of the couple $\np{\aggreg, \factor}$.
\end{mydef}

\subsection{Weak Time Consistency}
\label{subsection_WTC}

\begin{mydef}[Weak Time Consistency]\label{def_WTC}
The couple aggregator-factor
$\np{\aggreg, \factor}$ in~\eqref{eq:factor_aggregator}
is said to satisfy \emph{Weak Time Consistency (WTC)} if we have
    \begin{equation}
\label{weak_time_consistency_representation}
            \factor\np{\tailelement} = \factor\np{\tailelement'} 
            \Rightarrow 
            \aggreg\np{\headelement,\tailelement} = \aggreg\np{\headelement,\tailelement'}
            \eqsepv
            \forall \headelement \in \HeadSet
            \eqsepv \forall \np{\tailelement,\tailelement'} \in \TailSet^{2}
            \eqfinp
    \end{equation}
\end{mydef}

Here is our main result where we characterize the WTC property 
in terms of the subaggregator in~\eqref{definition_subaggregator}.

\begin{theo}[Nested decomposition of WTC mappings]
\label{WTC_theorem}
The couple  aggregator-factor
$\np{\aggreg, \factor}$ in~\eqref{eq:factor_aggregator}
is WTC if and only if the subaggregator set valued mapping~$\subaggreg$ 
in~\eqref{definition_subaggregator} is a mapping.
    In that case, the following \emph{Nested Formula} between mappings
holds true:
    \begin{equation}
\label{equation_proof_weak_nested_decomposition}
        \aggreg\np{\headelement,\tailelement}    = 
\subaggreg\bp{\headelement,\factor\np{\tailelement}}
        \eqsepv \forall \headelement \in \HeadSet
        \eqsepv \forall \tailelement \in \TailSet
    \eqfinp
    \end{equation}
\end{theo}

\begin{proof}
    Note that we always have by Equation~\eqref{definition_subaggregator} that 
    \begin{equation}\label{eq_inclusion_AinS}
        \aggreg\np{\headelement,\tailelement}
        \in
        \subaggreg\bp{\headelement, \factor\np{\tailelement}}
        \eqfinp
    \end{equation}
    \begin{enumerate}
         \item 
We suppose that the couple $\np{\aggreg, \factor}$ is Weak Time Consistent. 
Consider $\np{\headelement, \secondvalueelement}$ fixed in $\HeadSet \times 
         \image{\factor}$. 
We are going to show that the set valued mapping $\subaggreg$ 
is in fact a mapping, by proving that the set $\subaggreg\np{\headelement, 
\secondvalueelement}$, defined in~\eqref{definition_subaggregator}, 
is reduced to a singleton.
We consider two elements 
$\valueelement = \aggreg\np{\headelement, \tailelement}$ 
and $\valueelement' = \aggreg\np{\headelement, \tailelement'}$ 
in the set~$\subaggreg\np{\headelement, \secondvalueelement}$. 
By definition~\eqref{definition_subaggregator}, we have 
$\factor\np{\tailelement} = \factor\np{\tailelement'} 
= \secondvalueelement$. 
Then, using the Weak Time Consistency 
property~\eqref{weak_time_consistency_representation}, we deduce 
$\aggreg\np{\headelement,\tailelement} = \aggreg\np{\headelement,\tailelement'}$. Thus, $\subaggreg\np{\headelement,\secondvalueelement}$ is 
reduced to one value for $\secondvalueelement \in \image{\factor}$. 
The set valued mapping $\subaggreg$ is thus a mapping and, using 
Equation~\eqref{eq_inclusion_AinS}, we obtain 
$\aggreg\np{\headelement,\tailelement} = 
\subaggreg\bp{\headelement,\factor\np{\tailelement}}$.
\item 
We suppose now that the set valued mapping~$\subaggreg$, 
defined in~\eqref{definition_subaggregator}, is a mapping. 
Since~$\subaggreg$ is a mapping, we deduce by Equation~\eqref{eq_inclusion_AinS} that 
$\aggreg\np{\headelement,\tailelement} = 
\subaggreg\bp{\headelement,\factor\np{\tailelement}}$ for all 
$\tailelement \in \TailSet$. 
Therefore, we have the implications: 
$\factor\np{\tailelement} = \factor\np{\tailelement'} 
\Rightarrow \subaggreg\bp{\headelement,\factor\np{\tailelement}} 
=\subaggreg\bp{\headelement,\factor\np{\tailelement'}} 
\Rightarrow \aggreg\np{\headelement,\tailelement} = 
\aggreg\np{\headelement,\tailelement'}$.
We conclude that the weak time 
consistency property~\eqref{weak_time_consistency_representation} is satisfied.
     \end{enumerate}
 In both cases, we have shown that 
Equation~\eqref{equation_proof_weak_nested_decomposition} holds true.
\end{proof}

\begin{ex}[The couple 
    $\vardelim{\textrm{AV@R}_{\beta}\nc{\cdot + \cdot}\eqsepv \textrm{AV@R}_{\beta}\nc{\cdot \mid \tribu{F}}}$ 
    is not Weak Time Consistent]

We now give an example inspired from 
\cite*[Sect.~5.3.2, p.~188]{Pflug-Pichler:2014} and involving the well known 
Average Value at Risk. It helps to illustrate our main result and the notions we have introduced so far.

Let $\omeg = \np{\omega_{1},\omega_{2}, \omega_{3}, \omega_{4}}$, that 
we equip with the uniform probability distribution \( \prbt = 
\frac{1}{4}\delta_{\omega_{1}}+ \frac{1}{4}\delta_{\omega_{2}}+
\frac{1}{4}\delta_{\omega_{3}}+ \frac{1}{4}\delta_{\omega_{4}} \).

We introduce the sets $\HeadSet = \TailSet = \RR^{|\omeg|} = \RR^{4}$.
On this finite space $\omeg$, the Average Value at Risk of level $\beta$ ($0 \leq \beta \leq 1$) of a random variable $\va{X} : \omeg \to \RR$ 
is defined by \cite*{Rockafellar-Uryasev:2000}
\begin{equation}
    \textrm{AV@R}_{\beta}
    \np{\va{X}}
    =
    \min_{\alpha \in \RR} \quad
    \left \{
    \alpha +
    \frac{1}{1-\beta}
    \besp{\prbt}{\nc{\va{X} - \alpha}^{+}}
    \right \}
    \eqfinp
\end{equation}
Let $\tribu{F} =\ba{\emptyset, 
\na{\omega_{1}, \omega_{2}}, \na{\omega_{3}, \omega_{4}}, \omeg}$
be a $\sigma$-field on the space $\omeg$.
The Conditional Average Value at Risk of level $\beta$, 
of a random variable $\va{X} : \omeg \to \RR$ with 
respect to the $\sigma$-field $\tribu{F}$ is defined 
by (\cite*{ruszczynski2010risk}, Example 3):
\begin{align}
    \textrm{AV@R}_{\beta}
    \np{\va{X} \mid \tribu{F}}
    &=
    \inf_{\va{U} \preceq \tribu{F}}
    \Bga{
        \va{U} + \frac{1}{1-\beta}
        \bespc{\prbt}{\nc{\va{X} - \va{U}}^{+}}{\tribu{F}}
    }
    \eqfinv
\end{align}
where the infimum is understood point-wise among all random variables $\va{U}$ that are $\tribu{F}$-measurable, and where the level $\beta$ may be an $\tribu{F}$-measurable function with values in an interval 
$\nc{\beta_{\min},\beta_{\max}} \subset [0,1)$.

We define two mappings
\begin{subequations}
    \begin{align}
        \aggreg:\ &\HeadSet \times \TailSet \to \RR 
        &\factor:\ &\TailSet \to \RR^{2}
        \\
        &\np{\headelement, \tailelement} \mapsto \textrm{AV@R}_{0.5}\nc{\headelement + \tailelement}
        \eqfinv
        &
        &\tailelement \mapsto \textrm{AV@R}_{0.5}\nc{\tailelement \mid \tribu{F}}
        \eqfinp
   \end{align}
\label{eq:factor_aggregator_AVaR}
\end{subequations}
Consider four elements: 
a head $\headelement_{0} = \np{0,0,0,0} \in \HeadSet$, 
a first tail $\tailelement_{0} = \np{3,3,2,1} \in \TailSet$, 
a second tail $\tailelement_{0}' = \np{1,3,2,2} \in \TailSet$ 
and an element of the factor's image $\secondvalueelement_{0} = \np{3,2} \in \secondvalueset$.
On the one hand, the elements $\factor\np{\tailelement_{0}}$ 
and $\factor\np{\tailelement'_{0}}$ are equal, because
            \begin{equation}
                \textrm{AV@R}_{0.5}\nc{\tailelement_{0}|\tribu{F}} = 
                \underbrace{
                    \np{3;2}
                       }_{\secondvalueelement_{0}}
                = 
                \textrm{AV@R}_{0.5}\nc{\tailelement'_{0}|\tribu{F}}
                \eqfinp
            \end{equation} 
On the other hand, the elements 
$\aggreg\np{\headelement_{0}, \tailelement_{0}}$
and $\aggreg\np{\headelement_{0}, \tailelement'_{0}}$ are not equal, because
        \begin{equation}
3 = \textrm{AV@R}_{0.5}\nc{\headelement_{0} + \tailelement_{0}}
\neq \textrm{AV@R}_{0.5}\nc{\headelement_{0} + \tailelement'_{0}} = 2.5
            \eqfinp
        \end{equation}
        The subaggregator $\subaggreg$ in~\eqref{definition_subaggregator} is not a mapping since 
\begin{equation}
    \subaggreg\np{\headelement_{0},\secondvalueelement_{0}}
    =
    \ba{\textrm{AV@R}_{0.5}\nc{\headelement_{0} + \tailelement} 
        \mid \textrm{AV@R}_{0.5}\nc{\tailelement \mid \tribu{F}} = \secondvalueelement_{0} }
    \supset
    \na{2.5;\  3}
    \eqfinv
\end{equation}
and therefore the couple $\np{\aggreg, \factor}$ in~\eqref{eq:factor_aggregator_AVaR}
is not Weak Time Consistent.
\end{ex}

\subsection{Extensions to Usual and Strong Time Consistency} 
\label{subsection_USTC}  
With additional order structures on the image sets~$\valueset$
and $\secondvalueset$ of the {aggregator}~$\aggreg$ 
and of the {factor}~$\factor$, and possibly on the head set~$\HeadSet$ 
--- all presented in~\eqref{eq:factor_aggregator} --- 
we define two additional notions of Time Consistency, usual and strong.

\subsubsection{Usual Time Consistency (UTC)}\label{section_UTC}

Suppose that the image sets $\valueset$ and $\secondvalueset$ 
are equipped with orders, denoted by $\leq$.

\begin{mydef}[Definition of Usual Time Consitency]\label{definition_usual_time_consistency}
 The couple  aggregator-factor
$\np{\aggreg, \factor}$ in~\eqref{eq:factor_aggregator}
is said to satisfy \emph{Usual Time Consistency (UTC)} if we have
    \begin{equation}
\label{time_consistency_representation}
            \factor\np{\tailelement} \leq \factor\np{\tailelement'} 
            \Rightarrow 
            \aggreg\np{\headelement,\tailelement} \leq 
\aggreg\np{\headelement,\tailelement'}
            \eqsepv
            \forall \headelement \in \HeadSet
            \eqsepv \forall \np{\tailelement,\tailelement'} 
\in \TailSet^{2}
            \eqfinp
    \end{equation}
\end{mydef}

We extend the result of Theorem~\ref{WTC_theorem} as follows.

\begin{prop}[Nested decomposition of UTC mappings]
\label{decomposition_time_consistent}
The couple $\np{\aggreg, \factor}$ in~\eqref{eq:factor_aggregator} 
is UTC if and only if the set valued mapping $\subaggreg$ 
in~\eqref{definition_subaggregator} is a mapping 
and is \emph{increasing\footnote{Let $\mathbb{X}$ and $\mathbb{Y}$ be sets endowed with orders denoted by $\leq$.
    A mapping $M : \mathbb{X} \to \mathbb{Y}$ is said to be \emph{increasing} if
 \(        x \leq x' \Rightarrow M\np{x} \leq M\np{x'} \). } in its second argument}. 
In that case, 
the Nested Formula~\eqref{equation_proof_weak_nested_decomposition}
holds true. 
\end{prop}

The proof is left to the reader as it follows 
the proof of Theorem~\ref{WTC_theorem} with small variations.

\subsubsection{Strong Time Consistency (STC)}

Suppose that the head set~$\HeadSet$
and the image sets~$\valueset$ and $\secondvalueset$ 
are equipped with orders, denoted by $\leq$.

\begin{mydef}[Definition of Strong Time Consistency]
    The couple $\np{\aggreg, \factor}$ in Equation~\eqref{eq:factor_aggregator} 
is said to satisfy \emph{Strong Time Consistency} (STC) if we have    
\begin{equation}
\label{strong_time_consistency_representation}
            \left .
            \begin{array}{rcl}
                \factor\np{\tailelement} &\leq& \factor\np{\tailelement'}\\
                \headelement &\leq& \headelement'
            \end{array}
            \right \}
            \Rightarrow 
            \aggreg \np{\headelement,\tailelement} \leq 
\aggreg\np{\headelement',\tailelement'}
            \eqsepv 
            \forall \np{\headelement, \headelement', 
\tailelement,\tailelement'} \in \HeadSet^{2} \times \TailSet^{2}
            \eqfinp
    \end{equation}
\end{mydef}

We extend the results of Theorem~\ref{WTC_theorem} as follows.

\begin{prop}[Nested decomposition for STC mappings]
\label{proposition_characterization_STC_sublevel}
The couple $\np{\aggreg, \factor}$ in~\eqref{eq:factor_aggregator} 
is STC if and only if the set valued mapping $\subaggreg$ 
is a mapping \emph{increasing in its first and second arguments}. 
In that case, 
the Nested Formula~\eqref{equation_proof_weak_nested_decomposition}
holds true. 
\end{prop}

The proof is left to the reader as it follows 
the proof of Theorem~\ref{WTC_theorem} with small variations.

\subsubsection{Summing up results about WTC, UTC and STC}

In~\S\ref{subsection_WTC} and~\S\ref{subsection_USTC}, 
we have introduced three notions of Time Consistency, from the weakest to the
strongest. Of course, we have that a Strong Time Consistent couple is also Usual Time 
Consistent, and that a Usual Time Consistent couple is also Weak Time Consistent.
We sum up the different definitions and results in Table~\ref{table_time_consistency}.

\begin{table}[!ht]
    \centering
        \footnotesize
        \begin{tabular}{|c|c|c|c|}
        \hline
        &\multicolumn{3}{c|}{
            Weak~\eqref{weak_time_consistency_representation}
            \phantom{iiiiiii}$\Leftarrow$\phantom{iiiiiiiiiii} 
            Usual~\eqref{time_consistency_representation}
            \phantom{iiiiiiiiiii}$\Leftarrow$\phantom{iiiiiii}
            Strong~\eqref{strong_time_consistency_representation}}\\
        \hline
        \hline
        Definition
        &
        $
         \begin{aligned}
             \factor\np{\tailelement} &=  \factor\np{\tailelement'}\\
            &\Downarrow\\
            \aggreg\np{\headelement, \tailelement} 
            &= \aggreg\np{\headelement,\tailelement'}
        \end{aligned}
        $
        &
        $
        \begin{aligned}
            \factor\np{\tailelement} &\leq  \factor\np{\tailelement'}\\
            &\Downarrow\\
            \aggreg\np{\headelement, \tailelement} 
            &\leq \aggreg\np{\headelement,\tailelement'}
        \end{aligned}
        $
        &
        $
       \begin{aligned}
           \headelement &\leq  \headelement' \eqfinv\\
           \factor\np{\tailelement} &\leq  \factor\np{\tailelement'}\\
            &\Downarrow\\
            \aggreg\np{\headelement, \tailelement} 
            &\leq \aggreg\np{\headelement',\tailelement'}
        \end{aligned}
        $
        \\
        \hline
        $
        \begin{array}{c}
            \textrm{Characterization}\\
            \textrm{in terms of}\\
            \textrm{subaggregator}
        \end{array}
        $
        &
        \pbox{6cm}{$\subaggreg$ is a mapping}
        &
        $
        \begin{array}{c}
            \subaggreg \textrm{ is a mapping}\\
            \textrm{increasing}\\
            \textrm{in its second argument}
        \end{array}
        $
        &
        $
        \begin{array}{c}
            \subaggreg \textrm{ is a mapping}\\
            \textrm{increasing}\\
            \textrm{in both arguments}
        \end{array}
        $
        \\
        \hline
    \end{tabular}
    \caption{Characterization of Time Consistency in terms of subaggregator
\label{table_time_consistency}}
\end{table}

\subsection{Analytical properties of time consistent mappings}
\label{subsection_properties}

Here, we study properties inherited by the subaggregator~$\subaggreg$
in~\eqref{definition_subaggregator} when it is a mapping, that is, 
when the couple $\np{\aggreg, \factor}$ is Weak Time Consistent (see
Theorem~\ref{WTC_theorem}). 
 We insist that, in this part, we study how properties of 
    the subaggregator~$\subaggreg$ can be deduced from properties 
of aggregator~$\aggreg$ and factor~$\factor$.
Thus, our approach differs from other approaches in the literature,
like     \cite{Ruszczynski-Shapiro:2006},
    where properties of $\aggreg$ are deduced from properties of 
    $\subaggreg$ and $\factor$.
We focus on monotonicity, 
continuity, convexity, positive homogeneity and translation invariance.

\subsubsection{Monotonicity}

We suppose that the head set~$\HeadSet$,
the tail set $\TailSet$,
and the image sets~$\valueset$ and $\secondvalueset$ 
--- all presented in~\eqref{eq:factor_aggregator} --- 
are equipped with orders, denoted by $\leq$.
The proof of the following proposition is left to the reader 
as a direct application of the Nested Formula~\eqref{equation_proof_weak_nested_decomposition}.

\begin{prop}[Monotonicity]
    Let the couple $\np{\aggreg, \factor}$ be Weak Time Consistent,
    as in Definition~\ref{def_WTC}.
    If the mapping $\aggreg$ is increasing in its first argument, 
    then the subaggregator $\subaggreg$ in~\eqref{definition_subaggregator} is increasing in its first argument.
\end{prop}

\subsubsection{Continuity}
We suppose that the head set~$\HeadSet$,
the tail set $\TailSet$,
and the image sets~$\valueset$ and $\secondvalueset$ 
are metric spaces.

\begin{prop}[Continuity]\label{prop_continuity_subaggreg}
    Let the couple $\np{\aggreg, \factor}$ be Weak Time Consistent,
    as in Definition~\ref{def_WTC}.
    Assume that the tail set~$\TailSet$ is compact.
    If the factor $\factor$ is continuous
    and if the aggregator $\aggreg$ is continuous with a compact
    image~$\mathrm{Im}\np{\aggreg}=\aggreg\np{\HeadSet \times \TailSet}$,
    then the subaggregator $\subaggreg$ in~\eqref{definition_subaggregator}
    is continuous on \( \HeadSet \times \mathrm{Im}\np{\factor} \). 
\end{prop}

\begin{proof}
We prove the continuity of the subaggregator $\subaggreg$ on \( \HeadSet \times
    \image{\factor} \) 
by using the sequential characterization of the continuity on metric spaces.
For this purpose, we consider, on the one hand, 
$\np{\bar{\headelement}, \bar{\secondvalueelement}}$ 
element of $\HeadSet \times \image{\factor}$ and, on the other hand,
$\np{\headelement_{n}}_{n \in \mathbb{N}}$ 
    a sequence of elements of $\HeadSet$ converging to $\bar{\headelement}$
    and $\np{\secondvalueelement_{n}}_{n \in \mathbb{N}}$
    a sequence of elements of $\image{\factor}$ converging to
    $\bar{\secondvalueelement}$.
   We will show that $\subaggreg\np{\headelement_{n},\secondvalueelement_{n}}$ 
    converges to $\subaggreg\np{\bar{\headelement}, \bar{\secondvalueelement}}$.
    We introduce the notation $\mathcal{L}\bp{\na{u_{n}}}$ to denote the set of limit points 
    of a sequence $\np{u_{n}}_{n \in \mathbb{N}}$.

As \( \secondvalueelement_{n} \in \image{\factor} \), 
    there exists an element $\tailelement_{n} \in \TailSet$ such that 
    $\factor \np{\tailelement_{n}} = \secondvalueelement_{n}$ for each~$n$.
By the Nested Formula~\eqref{equation_proof_weak_nested_decomposition}, we
deduce that
\begin{equation}\label{eq_proof_continuity_1}
  \aggreg\np{\headelement_{n},\tailelement_{n}} 
    = \subaggreg\bp{\headelement_{n},\factor\np{\tailelement_{n}}} 
    =\subaggreg\bp{\headelement_{n},\secondvalueelement_{n}} \eqfinp
\end{equation}
We will now show that the set $\mathcal{L}\bp{ \ba{
    \aggreg\np{\headelement_{n},\tailelement_{n}}}}$
of limit points is reduced to the singleton 
$ \{ \subaggreg\np{\bar{\headelement}, \bar{\secondvalueelement}} \} $.
The proof is in several steps as follows:
\begin{enumerate}
\item 
$\mathcal{L}\bp{ \ba{ \aggreg\np{\headelement_{n},\tailelement_{n}}}} \not =
\emptyset$,
\item 
\( \mathcal{L}\bp{\ba{\aggreg\np{\headelement_{n},\tailelement_{n}}}}
        \subset
        \aggreg\bp{\bar{\headelement}, \mathcal{L}\np{\na{\tailelement_{n}}}} \),
\item 
\( \aggreg\bp{\bar{\headelement}, \mathcal{L}\np{\na{\tailelement_{n}}}} \) is
reduced to  the singleton 
$ \{ \subaggreg\np{\bar{\headelement}, \bar{\secondvalueelement}} \} $,
\item 
$\mathcal{L}\bp{ \ba{ \aggreg\np{\headelement_{n},\tailelement_{n}}}} =
\{ \subaggreg\np{\bar{\headelement}, \bar{\secondvalueelement}} \} $.
\end{enumerate}
 Here is the proof.

 \begin{enumerate}
    \item 
As the sequence $\bp{ \aggreg\np{\headelement_{n},\tailelement_{n}}}_{n \in
  \mathbb{N}}$
takes value in the compact set $\mathrm{Im}\np{\aggreg}$, we have that 
    $\mathcal{L}\bp{ \ba{ \aggreg\np{\headelement_{n},\tailelement_{n}}}} \neq
    \varnothing$.
    
    \item 
We prove that
    \(
        \mathcal{L}\bp{\ba{\aggreg\np{\headelement_{n},\tailelement_{n}}}}
        \subset
        \aggreg\bp{\bar{\headelement}, \mathcal{L}\np{\na{\tailelement_{n}}}}
            \).
    Let $a$ be an element of the set $\mathcal{L}\bp{\na{\aggreg\np{\headelement_{n},\tailelement_{n}}}}$.
    By definition of this latter set, there exists a subsequence
    $\bp{\aggreg\np{\headelement_{\Phi\np{n}},\tailelement_{\Phi\np{n}}}}_{n \in \NN}$
    converging to~$a$.
Now, we know that $\np{\headelement_{\Phi\np{n}}}_{n \in \NN}$ converges
to~$\bar{\headelement}$, but it
is not necessarily the case that $\np{\tailelement_{\Phi\np{n}}}_{n \in
\NN}$ converges. However, by compacity of the tail set $\TailSet$, 
there exist a subsequence $\np{t_{\Psi\circ\Phi\np{n}}}_{n \in \NN}$ 
    converging to a certain~$\bar{t} \in \mathcal{L}\bp{\na{t_{n}}}$.
    As the sequence $\bp{\aggreg\np{\headelement_{\Phi\np{n}},\tailelement_{\Phi\np{n}}}}_{n \in \NN}$
    is converging to~$a$, 
    the subsequence $\bp{\aggreg\np{\headelement_{\Psi\circ\Phi\np{n}},\tailelement_{\Psi\circ\Phi\np{n}}}}_{n \in \NN}$
    is also converging to~$a$.
Now that both inner subsequences converge, we use the continuity of the
mapping~$\aggreg$, and obtain that
    \(
    a =
    \lim_{n \to \infty}  \aggreg\np{\headelement_{\Psi\circ\Phi\np{n}},\tailelement_{\Psi\circ\Phi\np{n}}}
    =
    \aggreg\np{\bar{\headelement}, \bar{\tailelement}}
    \in
    \aggreg\bp{\bar{\headelement}, \mathcal{L}\np{\na{\tailelement_{n}}}}
    \).
    \item 
We prove the equality
    \( \aggreg\bp{\bar{\headelement}, \mathcal{L}\np{\na{\tailelement_{n}}}} 
    =
    \{ \subaggreg\np{\bar{\headelement}, \bar{\secondvalueelement}} \} \).
 Since the set $\mathcal{L}\bp{\na{\tailelement_{n}}}$ is not empty by compactness of $\TailSet$,
 we consider $\np{\bar{\tailelement}, \bar{\tailelement}'} \in \mathcal{L}\bp{\na{\tailelement_{n}}}^2$ 
    any two limits points of the sequence $\np{\tailelement_{n}}_{n \in
      \mathbb{N}}$.
As    $\factor \np{\tailelement_{n}} = \secondvalueelement_{n}$ and 
\( \lim_{n \to \infty} \secondvalueelement_{n} = \bar{\secondvalueelement} \),
we deduce that $\factor\np{\bar{\tailelement}} = \bar{\secondvalueelement} = \factor\np{\bar{\tailelement}'}$,
    by continuity of the factor mapping~$\factor$.
    The Nested Formula~\eqref{equation_proof_weak_nested_decomposition} gives
    \begin{equation*}
        \aggreg\np{\bar{\headelement},\bar{\tailelement}} 
        =
        \subaggreg\bp{\bar{\headelement}, \factor\np{\bar{\tailelement}}}
        =
        \subaggreg\np{\bar{\headelement}, \bar{\secondvalueelement}}
        =
        \subaggreg\bp{\bar{\headelement}, \factor\np{\bar{\tailelement}'}}
        =
        \aggreg\np{\bar{\headelement},\bar{\tailelement}'}
        \eqfinp
    \end{equation*}
This proves that \( \aggreg\bp{\bar{\headelement}, \mathcal{L}\np{\na{\tailelement_{n}}}} 
    =
    \{ \subaggreg\np{\bar{\headelement}, \bar{\secondvalueelement}} \} \).
    \item 
Gathering up the previous results, we obtain that
\begin{equation}
     \varnothing  \neq  \mathcal{L}\bp{ \ba{
         \aggreg\np{\headelement_{n},\tailelement_{n}}}}
\subset \aggreg\bp{\bar{\headelement}, \mathcal{L}\np{\na{\tailelement_{n}}}} 
    =
    \{ \subaggreg\np{\bar{\headelement}, \bar{\secondvalueelement}} \} \eqfinp
\end{equation}
We conclude that \( 
     \mathcal{L}\bp{ \ba{\aggreg\np{\headelement_{n},\tailelement_{n}}}}
     =
     \{ \subaggreg\np{\bar{\headelement}, \bar{\secondvalueelement}} \} \).
  \end{enumerate}
From Equation~\eqref{eq_proof_continuity_1}, we have the equalities 
    \( 
    \mathcal{L}\bp{\ba{\subaggreg\bp{\headelement_{n},\secondvalueelement_{n}}}}
    =
    \mathcal{L}\bp{ \ba{\aggreg\np{\headelement_{n},\tailelement_{n}}}}
    \)
    \(
    =
    \{ \subaggreg\np{\bar{\headelement}, \bar{\secondvalueelement}} \}
    \eqfinp
     \)
Therefore, the sequence $\subaggreg\np{\headelement_{n},\secondvalueelement_{n}}$ 
    converges to $\subaggreg\np{\bar{\headelement}, \bar{\secondvalueelement}}$.
This ends the proof.
\end{proof}

\subsubsection{Convexity}
\label{prop_convexity_subaggreg}

As we are dealing with convexity property, 
we assume that the sets $\HeadSet, \TailSet$ and $\secondvalueset$
in~\eqref{eq:factor_aggregator} are vector spaces. 
We also suppose that the aggregator \(  \aggreg: \HeadSet \times \TailSet \to \valueset \)
in~\eqref{eq:factor_aggregator} takes extended real values, that is, $\valueset = \RR \cup \na{-\infty, +\infty}$.

\begin{prop}\label{inheritance_convexity}
    Let the couple $\np{\aggreg, \factor}$ be Weak Time Consistent,
    as in Definition~\ref{def_WTC}.
    If there exists a nonempty convex subset $\bar{\TailSet} \subset \TailSet$ 
    such that $\factor\np{\bar{\TailSet}} = \mathrm{Im}\np{\factor}$ 
    and that the restricted function $\factor_{|\bar{\TailSet}}$ is affine, 
    and if the aggregator $\aggreg$ is jointly convex,
    then the subaggregator $\subaggreg$ in~\eqref{definition_subaggregator} is jointly convex on \( \HeadSet \times \mathrm{Im}\np{\factor} \).
\end{prop}

Before entering the proof, let us stress the point that,
even if the assumption that the restricted function $\factor_{|\bar{\TailSet}}$
be affine 
may look strong, it is quite realistic and widespread.
Indeed, for example, if the factor mapping $\factor$ is the identity mapping
on~$\secondvalueset$, then it satisfies the conditions of Proposition~\ref{inheritance_convexity}:
Conditional expectation or Conditional Average Value at Risk are hence encompassed in this framework.

\begin{proof}
    We introduce the notation $\mathrm{epi}\np{M}$ to denote the epigraph
    \footnote{Let $\mathbb{X}$ be a set.
        The epigraph of the mapping $M : \mathbb{X} \to \RR \cup \na{-\infty, +\infty}$ is defined by 
        $\mathrm{epi}\np{M} = \ba{\np{x,y} \in \mathbb{X} \times \mathbb{R} : M\np{x} \leq y}$
        where $y$ is a real number.
    }  of a mapping $M$.
    We prove that the subaggregator $\subaggreg$ is jointly convex by showing that
    its epigraph is jointly convex.
    
    Let $\bp{\np{\headelement_{1}, \secondvalueelement_{1}}, \valueelement_{1}}$
    and $\bp{\np{\headelement_{2}, \secondvalueelement_{2}}, \valueelement_{2}}$
    be two elements of the epigraph $\mathrm{epi}\np{\subaggreg}$ of the subaggregator. 
    We consequently have
    \(
        \valueelement_{1} \geq \subaggreg\np{\headelement_{1}, \secondvalueelement_{1}}
        \textrm{ and }
        \valueelement_{2} \geq \subaggreg\np{\headelement_{2}, \secondvalueelement_{2}}
    \)
    which by addition to
    \begin{equation}\label{eq_epigraph}
        \lambda \valueelement_{1}
        +\np{1-\lambda}\valueelement_{2}
        \geq
        \lambda \subaggreg\np{\headelement_{1}, \secondvalueelement_{1}}
        +\np{1-\lambda} \subaggreg\np{\headelement_{2}, \secondvalueelement_{2}}
        \eqfinv
    \end{equation}
    where $\lambda$ is an element of $\nc{0,1}$. 
    As, by assumption, $\factor\np{\bar{\TailSet}} = \mathrm{Im}\np{\factor}$, 
    there exist two elements $\np{\bar\tailelement_{1}, \bar\tailelement_{2}} \in {\bar\TailSet}^{2}$
    such that 
    \begin{equation}\label{equation_proof_convexity}
        \factor\np{\bar\tailelement_{1}} = \secondvalueelement_{1}
        \textrm{ and }
        \factor\np{\bar\tailelement_{2}} = \secondvalueelement_{2} \eqfinp 
    \end{equation}

    We have the succession of equalities and inequality 
    \begin{align*}
        \lambda \valueelement_{1}
        +\np{1-\lambda}\valueelement_{2}
        &\geq
        \lambda \subaggreg\np{\headelement_{1}, \secondvalueelement_{1}}
        +\np{1-\lambda} \subaggreg\np{\headelement_{2}, \secondvalueelement_{2}}
        \eqfinv
        \tag{by Eq.~\eqref{eq_epigraph},}
        \\
        &=
        \lambda \subaggreg\bp{\headelement_{1}, \factor\np{\bar{\tailelement}_{1}}} + \np{1-\lambda}\subaggreg\bp{\headelement_{2}, \factor\np{\bar{\tailelement}_{2}}}
        \eqfinv
        \tag{by Eq.~\eqref{equation_proof_convexity},}
        \\
        &=
        \lambda \aggreg\np{\headelement_{1}, \bar{\tailelement}_{1}} + \np{1-\lambda}\aggreg\np{\headelement_{2}, \bar{\tailelement}_{2}}
        \eqfinv
        \tag{by Eq.~\eqref{equation_proof_weak_nested_decomposition},}
        \\
        &\geq
        \aggreg\bp{\lambda \headelement_{1} + \np{1-\lambda}\headelement_{2},\lambda \bar{\tailelement}_{1} + 
            \np{1-\lambda}\bar{\tailelement}_{2}}
        \eqfinv
        \tag{by convexity of $\aggreg$,}
        \\
        &=
        \subaggreg\Bp{\lambda \headelement_{1} + \np{1-\lambda}\headelement_{2},\factor\bp{\lambda \bar{\tailelement}_{1} + 
            \np{1-\lambda}\bar{\tailelement}_{2}}}
        \eqfinv
        \tag{by Eq.~\eqref{equation_proof_weak_nested_decomposition},}
        \\
        &=
        \subaggreg\bp{\lambda \headelement_{1} + \np{1-\lambda}\headelement_{2},\lambda \factor\np{\bar{\tailelement}_{1}} + 
        \np{1-\lambda}\factor\np{\bar{\tailelement}_{2}}}
        \eqfinv
        \tag{by affinity of $\factor$ on $\bar{\TailSet}$,}
        \\
        &=
        \subaggreg\bp{\lambda \headelement_{1} + \np{1-\lambda}\headelement_{2},\lambda \secondvalueelement_{1} + \np{1-\lambda}\secondvalueelement_{2}}
        \eqfinv
        \tag{by Eq.~\eqref{equation_proof_convexity}.}
    \end{align*}    
    We deduce that the element $\Bp{\bp{\lambda \headelement_{1} + \np{1-\lambda}\headelement_{2},\lambda \secondvalueelement_{1} + \np{1-\lambda}\secondvalueelement_{2}},
        \lambda \valueelement_{1}+\np{1-\lambda}\valueelement_{2}}$ is in the epigraph $\mathrm{epi}\np{\subaggreg}$ of the subaggregator.
This ends the proof.
\end{proof}

% \modi{Detail a bit more the explanation and put it before proof}
% \st{For example, if $\factor$ is a projection of $\TailSet$ on $\secondvalueset$,
% that is, if $\factor^2 = \factor$, then $\factor$ is the identity mapping
% on~$\secondvalueset$,
% so that $\factor$ satisfies the conditions of
% Proposition}~\ref{inheritance_convexity}.
% \emod

Notice that, if the factor $\factor$ is only convex, we cannot conclude in general. 
For example, let $\aggreg\np{\headelement,\tailelement} = \headelement + \tailelement$
be an aggregator and
let $\factor\np{\tailelement} = \exp\np{\tailelement}$
be a factor. 
Then the couple $\np{\aggreg, \factor}$ is Weak Time Consistent
with an associated subaggregator 
$\subaggreg\np{\headelement,\secondvalueelement} = \headelement + \ln\np{\secondvalueelement}$
which is not convex.

\subsubsection{Homogeneity}
As we are dealing with homogeneity property, 
we assume that the sets $\HeadSet, \TailSet, \valueset$ and $\secondvalueset$
in~\eqref{eq:factor_aggregator} 
are endowed with an external multiplication with the scalar field $\RR$.

\begin{prop}[Positive homogeneity]
    Let the couple $\np{\aggreg, \factor}$ be Weak Time Consistent,
    as in Definition~\ref{def_WTC}.
    If the mapping $\aggreg$ is jointly positively homogeneous
    and if the mapping $\factor$ is positively homogeneous,
    then the subaggregator $\subaggreg$ is jointly positively homogeneous.
\end{prop}

\begin{proof}
    Let $\np{\headelement, \tailelement}$ be element of $\HeadSet \times \TailSet$.
    Let $\lambda \in \RR^{+}$.
    We have the following equalities
    \begin{align*}
        \subaggreg\bp{\lambda \headelement, \lambda \factor\np{\tailelement}}
        &=
        \subaggreg\bp{\lambda \headelement, \factor\np{\lambda\tailelement}}
        \eqfinv
        \tag{by positive homogeneity of $\factor$}
        \\
        &=
        \aggreg\np{\lambda \headelement, \lambda\tailelement}
        \eqfinv
        \tag{by the Nested Formula~\eqref{equation_proof_weak_nested_decomposition}}
        \\
        &=\lambda\aggreg\np{\headelement, \tailelement}
        \eqfinv
        \tag{by positive homogeneity of $\aggreg$}
        \\
        &=\lambda\subaggreg\bp{\headelement, \factor\np{\tailelement}}
        \eqfinv
        \tag{by the Nested Formula~\eqref{equation_proof_weak_nested_decomposition}.}
    \end{align*}
    This ends the proof.
\end{proof}

\subsubsection{Translation invariance}
As we are dealing with translation invariance, we assume that the sets
$\HeadSet, \TailSet, \valueset$ and $\secondvalueset$ 
in~\eqref{eq:factor_aggregator} are endowed with
an addition~$+$. We also assume that there exists a set $\mathbb{I}$ of invariants
which is a common subspace of $\HeadSet, \TailSet, \valueset$ and
$\secondvalueset$, as follows. 

\begin{mydef}\label{def_TI}
    Let $\mathbb{X}$ and $\mathbb{Y}$ be sets equipped with an addition $+$.
    Let $\mathbb{I} \subset \mathbb{X} \cap \mathbb{Y}$ be a common subset of $\mathbb{X}$ and $\mathbb{Y}$.
    A mapping $M : \mathbb{X} \to \mathbb{Y}$ is said to be
    \emph{$\mathbb{I}$-translation invariant} 
if
    \begin{equation}
        M\np{x+i} = M\np{x}+i
         \eqsepv
        \forall x \in \mathbb{X}
        \eqsepv
        \forall i \in \mathbb{I}
        \eqfinp
    \end{equation}
\end{mydef}

\begin{prop}
    Let the couple $\np{\aggreg, \factor}$ be Weak Time Consistent,
    as in Definition~\ref{def_WTC}.
    If the mapping~$\aggreg$ is jointly translation invariant
    and if the mapping~$\factor$ is translation invariant
    then the subaggregator $\subaggreg$ is jointly translation invariant.
\end{prop}

\begin{proof}
     Let $\np{\headelement, \tailelement}$ be an element of $\HeadSet \times \TailSet$.
     Let $i \in \mathbb{I}$.
    We have the following equalities:
    \begin{align*}
        \subaggreg\bp{\headelement+i, \factor\np{\tailelement}+i}
        &=
        \subaggreg\bp{\headelement+i, \factor\np{\tailelement+i}}
        \eqfinv
        \tag{by translation invariance of $\factor$}
        \\
        &=
        \aggreg\np{\headelement+i, \tailelement+i}
        \eqfinv
        \tag{by the Nested Formula~\eqref{equation_proof_weak_nested_decomposition}}
        \\
        &=\aggreg\np{\headelement, \tailelement}+i
        \eqfinv
        \tag{by translation invariance of $\aggreg$}
        \\
        &=\subaggreg\bp{\headelement, \factor\np{\tailelement}}+i
        \eqfinv
        \tag{by the Nested Formula~\eqref{equation_proof_weak_nested_decomposition}.}
    \end{align*}
     We conclude that the subaggregator~$\subaggreg$ is 
    jointly translation invariant.
\end{proof}

\section{Revisiting the literature}
\label{section_application}

In Sect.~\ref{section_literature_review}, we have gone through a selection of
papers, touching
Time Consistency and Nested Formula in various settings. 
In Sect.~\ref{section_results}, we have formally stated our (abstract) definitions
of Time Consistency (TC) and Nested Formula (NF), 
and we have proven their equivalence. We have also provided conditions to obtain analytical properties of the mapping $\subaggreg$ appearing
in the Nested Formula, such as monotonicity, 
continuity, convexity, positive homogeneity and translation invariance.

Now, we return to the literature 
that we have briefly reviewed in~Sect.~\ref{section_literature_review},
and we show how our framework applies.
For this purpose, we go through each article and try to answer two questions.

First, what are the core assumptions that relate to our minimal notions of 
Time Consistency or Nested Formula? In particular, 
what are the heads and the tails and how are the Time Consistency axiom or
the Nested Formula formulated?
We will recover that the various definitions in the selection appear 
as special cases of ours.

Second, what are the assumptions that are additional to the core TC or NF formulations, 
and what do they imply for the subaggregator in the Nested Formula?
We will extract the additional assumptions specific to each author
and hence highlight their additional contribution.

\subsection{Axiomatic for Time Consistency (TC)}

We start our survey with the group of authors stating Time Consistency axiomatic. This group is subdivided between economists, 
who deal with lotteries and preferences, 
and probabilists, who deal with stochastic processes
and dynamical risk measures.
 
\subsubsection{Lotteries and preferences}

Kreps and Porteus (\cite*{Kreps-Porteus:1978}, \cite*{Kreps-Porteus:1979}) 
state a temporal consistency axiom (Axiom 2.1) in the first paper. In the second paper, they focus on 
the particular case of two stage problems. Their axiomatic is an instance of our
Definition~\ref{definition_usual_time_consistency}
of Usual Time Consistency.
With our Proposition~\ref{decomposition_time_consistent}, we directly deduce the existence of
a subaggregator increasing in its second argument and a Nested Formula,
whereas they obtain a stronger result under stronger assumptions. 
Indeed, they add assumptions of continuity, 
substitution (related to convexity) and focus on Usual Time Consistency with strict 
inequalities.
This enables them to obtain a subaggregator 
which is continuous and strictly increasing in its second argument and is defined by
((Lemma 4, Theorem 2) and Proposition 1 respectively):
\( %\begin{equation}
    u_{y_{t}}:
    \ba{
        \np{z,\gamma} \in Z_{t}\times \RR:
        \gamma = U_{y_{t},z}\np{x}
        \textrm{ for some }
        x \in \XX_{t+1}
    }
    \to
    \RR
\). %    \eqfinp \end{equation}
\medskip

Epstein and Schneider \cite*{Epstein-Schneider:2003} state an axiom of 
Dynamic Consistency (Axiom 4: DC) which is a particular case of our 
Definition~\ref{definition_usual_time_consistency} of Usual Time Consistency.
With our Proposition~\ref{decomposition_time_consistent}, we directly deduce the existence of
a subaggregator increasing in its second argument and a Nested Formula,
whereas they obtain a stronger result under stronger assumptions. 
Indeed, they introduce four additional axioms --- Conditional Preferences 
(CP), Multiple Priors (MP), Risk Preference (RP) and Full 
Support (FS) --- that ensure a particular form of the subaggregator. 
MP and CP ensure that the subaggregator can be represented as a minimum of expectation over a rectangular set 
of probabilities which is closed and convex. MP and RP ensure that the criterion
is additive over time. 
FS ensures that the probability measures have full support.
Epstein and Schneider obtain the 
following Nested Formula\footnote{
The equation %~\eqref{equation_epstein_schneider} 
is the original transcription of the
formula in~\cite*{Epstein-Schneider:2003}, to which we refer the reader for a
better understanding. By laying it out, 
we only want to stress the Nested Formula between $V_{t}$ and $V_{t+1}$.
} associated to Time Consistency (Theorem 3.2):
\( %\begin{equation}\label{equation_epstein_schneider}
    V_{t}\np{h,\omega} 
    =
    \min_{m \in \cP_{t}^{+1}\np{\omega}}
    \int
    \Bc{
        u\bp{h_{t}\np{\omega}}
        +
        \beta
        V_{t+1}\np{h}
    }dm
\). %    \eqfinp \end{equation}

\subsubsection{Dynamic risk measures and processes}

Ruszczy\'nsky studies \cite*{ruszczynski2010risk} dynamic risk measures $\na{\rho_{s,T}}_{s=1}^{T}$. 
Time Consistency (his Definition~3),
appears as a particular case of our Usual Time Consistency 
Definition~\ref{definition_usual_time_consistency}.
With our Proposition~\ref{decomposition_time_consistent}, we directly deduce the existence of
a subaggregator increasing in its second argument and a Nested Formula,
whereas Ruszczy\'nsky obtains a stronger result under stronger assumptions. 
Indeed, he adds assumptions that induce a particular form for the subaggregator.
From a conditional risk measure $\rho_{s,T}$, he defines 
mappings $\rho_{s,s'}$ with $s \leq s' \leq T$.
With our notations for aggregator~$\aggreg$ and factor~$\factor$, he
then focuses on the case where the initial assessment is $\aggreg = \rho_{s,T}$ 
and the future assessment is $\factor = \rho_{s',T}$. 
With two additional assumptions of invariance 
by translation and normalization ($\rho_{s,T}\np{0} = 0)$, 
Ruszczy\'nsky is able to state that 
the subaggregator has the specific form (Theorem 1):
\( %\begin{equation}
    \subaggreg = \rho_{s,s'}
\). %    \eqfinp \end{equation}
\medskip

In~\cite*{ADEH-coherent:2007},
Artzner, Delbean, Eber, Heath and Ku present Time Consistency (their Definition 4.1) 
which appears as an instance of our Definition~\ref{definition_usual_time_consistency}
of Usual Time Consistency. 
With our Proposition~\ref{decomposition_time_consistent}, we directly deduce the existence of
a subaggregator increasing in its second argument and a Nested Formula,
whereas they obtain a stronger result under stronger assumptions. 
Indeed, they  study particular mappings of the form $ \Psi_{t} = \sup_{\prbt \in \cP}  
\nespc{\prbt}{\cdot}{\tribu{F}_{t}}$, 
where $\cP$ is a subset of probabilities and $\np{\tribu{F}_{t}}_{t = 0}^{T}$ is a filtration.
They make an intermediary step before presenting a Nested Formula. 
They use a tool that they name stability by pasting (rectangularity) 
of the set~ $\cP$ of probability distributions.
With our notations for aggregator~$\aggreg$ and factor~$\factor$, 
this enables them to obtain, for $s\leq s'$, that if $\aggreg=\Psi_{s}$ and $\factor = \Psi_{s'}$ 
then the subaggregator has the specific form (Theorem 4.2):
\( % \begin{equation}
    \subaggreg\np{\headelement, \cdot} = \Psi_{s}\np{h + \cdot}
\). %     \eqfinp \end{equation} 

\subsection{Axiomatic for Nested formulas (NF)}

Shapiro and Ruszczy\'nski \cite*{Ruszczynski-Shapiro:2006}  study 
a family of conditional risk mapping
$\rho_{t} = \rho_{\cX_{2}\mid\cX_{1}} 
\circ \cdots \circ \rho_{\cX_{t}\mid\cX_{t-1}}$  (Equation (5.8)).
Each $\rho_{t}$ is increasing and is associated with a $\sigma$-algebra $\tribu{F}_{t}$,
where $\np{\tribu{F}_{t}}_{t=2}^{T}$ is a filtration.
As these mappings~$\rho_{t}$ are instances of the mappings in our 
Nested Formula~\eqref{equation_proof_weak_nested_decomposition},
they are Usual Time Consistent,
by using our Proposition~\ref{decomposition_time_consistent}.
With our notations for aggregator~$\aggreg$ and factor~$\factor$, 
and with additional assumptions of monotonicity, translation invariance, convexity and homogeneity, 
Shapiro and Ruszczy\'nski obtain that, if the initial assessment is $\aggreg = \rho_{t}$ and the future 
assessment is $\factor = \rho_{t+1}$, then the subaggregator is (Theorem 5.1)
\( % \begin{equation}
    \subaggreg = \rho_{\cX_{t+1}\mid\cX_{t}}
 \). %     \eqfinp\end{equation}
\medskip

Shapiro (\cite*{Shapiro:2016}) focuses on a future assessment 
and on a subaggregator of the form (Definition 2.1)
\( % \begin{align*}
    \factor = \sup_{\prbt \in \cP} \bespc{\prbt}{\cdots    \sup_{\prbt \in \cP}  \nespc{\prbt}{\cdot}{\tribu{F}_{T-1}}    }{\tribu{F}_{0}}
    \eqsepv
    \subaggreg = \sup_{\prbt \in \cP} \nesp{\prbt}{\cdot}
 \). %       \eqfinp\end{align*} 
With our notations for aggregator~$\aggreg$ and factor~$\factor$, 
this Nested Formula is an instance of our Nested Formula~\eqref{equation_proof_weak_nested_decomposition}.
We can define a natural initial assessment which is Usual Time Consistent with
the future assessment, by using our Proposition~\ref{decomposition_time_consistent}.
With additional assumptions of finiteness, Shapiro obtains that there exists 
a bounded set~$\widehat{\cP}$ of probability distributions such that 
the initial assessment has the specific form (Theorem 2.1)
\( % \begin{equation}
    \aggreg = 
    \sup_{\prbt \in \widehat{\cP}} 
    \nesp{\prbt}{\cdot}   
 \). %      \eqfinp\end{equation}
Besides, with additional assumption (Theorem 2.2) that $\cP$ is convex, 
bounded and weakly closed, Shapiro establishes that $\cP = \hat{\cP}$.
\medskip

De Lara and Leclère \cite*{DeLara-Leclere:2016} study composition of one time step aggregators. 
They make a distinction between uncertainty aggregator and time step aggregator, 
and they write a Nested Formula (Equation~(11)) 
which is an instance of our Formula~\eqref{equation_proof_weak_nested_decomposition}.
We can naturally define an initial assessment from this composition operation
which is time consistent with the one time step aggregator, by using our 
Proposition~\ref{decomposition_time_consistent}.
They add an additional 
hypothesis of monotonicity and one of commutation between uncertainty aggregator and time aggregator. 
They deduce that the initial assessment can be defined as the composition 
between a one time step aggregator (subaggregator) and a future assessment (Theorem~9).

\section{Two classes of time consistent mappings}
\label{Two_classes_of_time_consistent_mappings}%{chapter_risk_measures}

In this section, we present two classes of mappings
%translation invariant mappings and Fenchel-Moreau conjugates, 
that display time consistency under suitable assumptions. 
%
% investigate the properties induced by 
% time consistency for two classes of mappings.
% Indeed, we are motivated by 
% the mappings describing risk measures
% but we remain  as general as possible. 
% In particular, we do not want to stick to an additive criterion, that is,
% we want to be able to encompass nested formula of the kind
% \begin{equation}
%     \nmes{0}{\va{X} \oplus \va{Y}}
%     =
%     \nmes{1}{\va{X} \oplus \nmes{2}{\va{Y}}}
%     \eqfinv
% \end{equation}
% where $\oplus$ is not necessarily the ``usual addition''.
%
% Two classes are studied in the world of risk measures.
% We investigate them in an abstract framework.
We study in~\sect{section_IT_group_bis} translation invariant mappings
motivated by the representation of risk measures in terms of acceptance set.
Then, in~\sect{section_TC_rectangular}, 
we study mappings that are defined as Fenchel-Moreau conjugates
motivated by the dual reformulation of convex risk measures.

\subsection{Time consistent translation invariant mappings}
\label{section_IT_group_bis}

We study translation invariant mappings defined on ordered groups. 
We associate to each such mapping
an acceptance set which is the level set of level 0.
We prove that time consistency between two translation invariant mappings
is equivalent to an inclusion between acceptance sets.

\subsubsection{Translation invariant mappings on a group}
\label{subsection_IT_group}

We provide here the definition of a translation invariant mapping
and the one of an acceptance set. 
With these notions, we will
state our contribution.
We first recall the definition of an ordered group.
\begin{mydef}
    The triplet $\np{\secondvalueset, \oplus, \leq}$ 
    is said to be an \emph{ordered group}
    if $\secondvalueset$ is a set,
    $\np{\secondvalueset, \oplus}$ is a group,
    $\np{\secondvalueset, \leq}$ is an ordered set,
    and the order $\leq$ is compatible with $\oplus$, i.e.
        \begin{equation}
            \secondvalueelement_{1} \leq \secondvalueelement_{2} 
            \Rightarrow
            \secondvalueelement_{1} \oplus \secondvalueelement_{3} 
            \leq
            \secondvalueelement_{2} \oplus \secondvalueelement_{3} 
            \eqsepv
            \forall 
            \np{\secondvalueelement_{1}, \secondvalueelement_{2}, \secondvalueelement_{3}}
            \in \secondvalueset^{3}
            \eqfinp
        \end{equation}
\end{mydef}

We now provide the definition of translation invariant mappings on a group.
\begin{mydef}\label{def_IT_RM}
    Let $\np{\TailSet, \oplus}$ be a commutative group
    and $\np{\secondvalueset, \oplus}$ be a subgroup of $\np{\TailSet, \oplus}$,
    that we denote by
    \begin{equation}
        \np{\secondvalueset, \oplus}
        \subset
        \np{\TailSet, \oplus}
        \eqfinp
    \end{equation}
    A \emph{$\np{\TailSet, \secondvalueset}$-translation invariant mapping}
    is a mapping $\factor : \TailSet \to \secondvalueset$ that satisfies
    \begin{equation}\label{eq_IT_ominus}
        \factor\np{\tailelement \oplus \secondvalueelement} 
        = 
        \factor\np{\tailelement} \oplus  \secondvalueelement
        \eqsepv \forall \tailelement \in \TailSet
        \eqsepv
        \forall \secondvalueelement \in \secondvalueset
        \eqfinp
    \end{equation}
    \label{translation-invariant}
    
    In addition, if $\np{\factor, \oplus, \leq}$ is an ordered group, 
    we introduce the notations $\cA_{\factor}$ and $\cA_{\factor \mid \secondvalueset}$
    to deal with particular level sets of the  $\np{\TailSet, \secondvalueset}$-translation invariant mapping $\factor: \TailSet \to \secondvalueset$:
    \begin{subequations}
        \begin{align}
            \cA_{\factor} 
            &= 
            \ba{\tailelement \in \TailSet \mid \factor\np{\tailelement} \leq 0}
            \eqfinv\label{eq_def_acceptance_set}
            \\
            \cA_{\factor_{\mid \secondvalueset}}
            &=
            \ba{\secondvalueelement \in \secondvalueset \mid \factor\np{\secondvalueelement} \leq 0}
            =
            \cA_{\factor} \cap \secondvalueset
            \eqfinp
            \label{eq_acceptance_set_subaggreg}
    \end{align}
    \end{subequations}
\end{mydef}

\subsubsection{Characterization of UTC in terms of acceptance sets}
Given two translation invariant mappings $\factor$ and $\mesu$
as in Definition~\ref{def_IT_RM}, 
we will build an aggregator $\aggreg_{\mesu}$ such that
the couple $\np{\aggreg_{\mesu}, \factor}$ is time consistent
as in Definition~\ref{definition_usual_time_consistency}.
\begin{figure}[ht]
    \centering
    \begin{tikzpicture}[->,>=stealth',shorten >=1pt,auto,node distance=2cm, semithick]
        \tikzstyle{every state}=[draw=none,text=black]

        \node[state]    (A)                                 {$\HeadSet$};
        
        \node[state]    (AB)  [right of=A,  xshift = -1.5cm]                   {$\times$};
        \node[state]    (B)  [right of=AB, xshift = -1.5cm]                   {$\TailSet$};

        \node[state]    (C)  [right of=B]                   {$\valueset$};
        \node[state]    (D)  [below of=A]                   {$\HeadSet$};
        \node[state]    (E)  [below of=AB]                   {$\times$};
        \node[state]    (F)  [below of=B]                   {$\secondvalueset$};
        
        \path   
        (B)     edge  [above]   node {$\aggreg_{\mesu}$}  (C)
                (B.090)     edge  [bend left=45]   node {$\mesu$}  (C.090)
                (B)     edge  [left]   node {$\factor$}  (F)
                (F)     edge  [right]  node {$\subaggreg$}  (C);
       
    \end{tikzpicture}
    \caption{Representation of links between mappings of Proposition~\ref{prop_sum_acceptance_set}}
\end{figure}
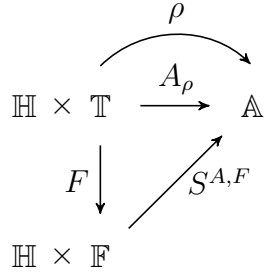

The next proposition is a generalization, with our notations, 
of Lemma 11.14 and Proposition 11.15 of \citet*{follmer2011stochastic},
since we do not refer to risk measures on probability spaces
but to more general sets.

\begin{prop}\label{prop_sum_acceptance_set}
    Let $\np{\TailSet, \oplus}$ be a commutative group.
    Given two subgroups
    \begin{equation}
        \np{\HeadSet, \oplus}
        \subset
        \np{\TailSet, \oplus}
        \textrm{ and }
        \np{\valueset, \oplus}
        \subset
        \np{\TailSet, \oplus}
        \eqfinv
    \end{equation}
    and a $\np{\TailSet, \valueset}$-translation invariant mapping $\mesu: \TailSet \to \valueset$,
    we define the mapping $\aggreg_{\mesu}: \HeadSet \times \TailSet \to \valueset$ by
    \begin{align}\label{eq_aggreg_mesu}
        \aggreg_{\mesu} : \HeadSet \times \TailSet &\to \np{\valueset, \oplus, \leq} \eqfinv\\
        \np{\headelement, \tailelement} &\mapsto \mesu\np{\headelement \oplus \tailelement}
        \eqfinp
    \end{align}
    Let $\factor: \TailSet \to \np{\secondvalueset, \oplus, \leq}$ 
    be a $\np{\TailSet, \secondvalueset}$-translation invariant mapping.
     If we have that
    \begin{itemize}
        \item $\np{\HeadSet, \oplus}
        \subset
        \np{\valueset, \oplus}
        \subset
        \np{\secondvalueset, \oplus}
        \subset
        \np{\TailSet, \oplus}$,
        \item the $\np{\TailSet,\valueset}$-translation invariant mapping $\mesu: \TailSet \to \valueset$ is increasing,
        \item the $\np{\TailSet,\secondvalueset}$-translation invariant mapping $\factor: \TailSet \to \secondvalueset$ satisfies $\factor\np{0} = 0$
        (where $0$ is the neutral element of $\np{\TailSet, \oplus}$),
    \end{itemize}
    then the couple of mappings $\np{\aggreg_{\mesu},\factor}$ is Usual Time Consistent 
    if and only if
    \begin{equation}\label{eq_eq_acc_set}
        \cA_{\factor} \oplus \cA_{\mesu_{\mid \secondvalueset}}
        =
        \cA_{\mesu}
        \eqfinv
    \end{equation}
    where $\cA_{\factor}$, $\cA_{\mesu \mid \secondvalueset}$ and $\cA_{\mesu}$
    are defined in~\eqref{eq_def_acceptance_set} and~\eqref{eq_acceptance_set_subaggreg}.
\end{prop}

\begin{proof}
    We refer the reader to Appendix~\ref{Appendix} for the proof.
\end{proof}

% We have provided a characterization of time consistency
% for translation invariant mappings in terms of acceptance sets 
% leading to 
Equation~\eqref{eq_eq_acc_set}
% \begin{equation}
%         \cA_{\factor} \oplus \cA_{\mesu_{\mid \secondvalueset}}
%         =
%         \cA_{\mesu}
%         \eqfinp
% \end{equation}
% This equation 
establishes a nice relation between acceptance sets 
of the original mapping $\rho$ and the ``conditional'' mapping $F$.
However, it remains difficult to solve when the variables
are the mappings $\mesu: \TailSet \to \valueset$ and $\factor: \TailSet \to \secondvalueset$ 
given in Proposition~\ref{prop_sum_acceptance_set} since it is an implicit equation in $\mesr$.

\subsection{Time consistent convex mappings}
\label{section_TC_rectangular}

Here, we focus on time consistency for mappings that
are defined as Fenchel-Moreau conjugates.
We are motivated by results on dual representation of convex risk mappings 
\citep*{ADEH-coherent:1999}. %\cite{ADEH-coherent:1999}.

We first recall Fenchel-Moreau conjugacy with general couplings
(not necessarily the classic duality pairing).
Then, we state our main theorem that provide a nested formula
and hence time consistency of mappings defined as Fenchel-Moreau conjugates.

\subsubsection{Basic tools to deal with Fenchel-Moreau conjugacies}

The formal tools of couplings and Fenchel-Moreau conjugates 
were introduced in the seminar paper of~\citet*{inf63sous}.
We recall that
$\overline{\RR} = \barRR = \RR \cup \na{-\infty, +\infty}$.

When we manipulate functions with values 
in~$\overline{\RR}$,
we adopt the Moreau \emph{lower addition} or
\emph{upper addition} defined in Equations~\eqref{eq:lower_addition} and~\eqref{eq:upper_addition}, 
depending on whether we deal with $\sup$ or $\inf$ operations. 
We only recall useful definitions to make the article self-contained.
In the sequel, $u$, $v$ and $w$ are any elements of~$\bar\RR$.
\medskip

The Moreau \emph{lower addition} and \emph{upper addition} extend the usual addition with 
\begin{subequations}
    \begin{equation}
        \np{+\infty} \LowPlus \np{-\infty} 
        = 
        \np{-\infty} \LowPlus \np{+\infty} = -\infty 
        \eqfinv
        \label{eq:lower_addition}
    \end{equation}
    \begin{equation}
        \np{+\infty} \UppPlus \np{-\infty} = 
        \np{-\infty} \UppPlus \np{+\infty} = +\infty \eqfinp
        \label{eq:upper_addition}
    \end{equation}
\end{subequations}
and they display the following properties:
\begin{subequations}
    \begin{equation}
        -(u \UppPlus v) = (-u) \LowPlus (-v) \eqsepv 
        -(u \LowPlus v) = (-u) \UppPlus (-v) \eqfinp 
        \label{eq:lower_upper_addition_minus}  
        \end{equation}
    \begin{equation}
        \sup_{a \in \mathbb{A}} f(a) 
        \LowPlus 
        \sup_{b \in \BB} g(b)
            =
        \sup_{a \in \mathbb{A}, b \in \BB} 
        \bp{f(a) \LowPlus g(b)} 
        \eqfinv
        \label{eq:lower_addition_sup}
    \end{equation}
\end{subequations}

\paragraph{Background on Fenchel-Moreau conjugacy with respect to a coupling.}
Let be given two sets $\PRIMAL$ and $\DUAL$.
Consider a \emph{coupling} function 
\( \coupling : \PRIMAL \times \DUAL \to \barRR \).
We also use the notation \( \PRIMAL 
\overset{\coupling}{\leftrightarrow} \DUAL \) for a coupling, so that
\begin{equation}\label{eq_coupling}
  \PRIMAL \overset{\coupling}{\leftrightarrow} \DUAL 
  \iff
  \coupling : \PRIMAL \times \DUAL \to \barRR \eqfinp
\end{equation}

\begin{mydef}\label{def_FM_conjugate}
  The \emph{Fenchel-Moreau conjugate} of a 
  function \( \fonctionprimal : \PRIMAL  \to \barRR \), 
  with respect to the coupling~$\coupling$ in~\eqref{eq_coupling}, is
  the function \( \SFM{\fonctionprimal}{\coupling} : \DUAL  \to \barRR \)
  defined by
  \begin{equation}
    \SFM{\fonctionprimal}{\coupling}\np{\dual} = 
    \sup_{\primal \in \PRIMAL} \Bp{ \coupling\np{\primal,\dual} 
      \LowPlus \bp{ -\fonctionprimal\np{\primal} } } 
\eqsepv \forall \dual \in \DUAL
\eqfinp
    \label{eq:upper-Fenchel-Moreau_conjugate}
  \end{equation}
\end{mydef}

\subsubsection{Main result: nested formula for Fenchel-Moreau conjugates}
\label{subsection_abstract_rectangular}

We provide a nested formula between mappings
defined as Fenchel-Moreau conjugates.
We introduce the notion of decomposable coupling.
\begin{mydef}\label{def_rectangle_abstract}
    Let $\XX$, $\YY$, $\ZZ$ and $\YY'$ be four sets
    and let $\theta_{\XX \times \ZZ}$, $\theta_{\ZZ}$ and $\theta_{\XX}$ 
    be three mappings with values in $\YY'$
    \begin{subequations}
    \begin{align}
        \theta_{\XX \times \ZZ}&: \XX \times \ZZ \to \YY'
        \eqfinv
        \\
        \theta_{\ZZ}&: \ZZ \to \YY'
        \eqfinv
        \\
        \theta_{\XX} &: \XX \to \YY'
        \eqfinp
    \end{align}
    \end{subequations}
    Let $\varphi: \YY' \times \YY \to \barRR$ be a coupling
    between $\YY'$ and $\YY$.
    
    We say that the coupling $\varphi$ is \emph{$\np{\theta_{\XX \times \ZZ}, \theta_{\XX}, \theta_{\ZZ}}$-decomposable} if
    \begin{align}
            \varphi\bp{\theta_{\XX}\np{x},y}
            &=
            \sup_{z \in \ZZ}
            \Ba{\varphi\bp{\theta_{\XX \times \ZZ}\np{x,z},y} \plusdot \Bp{-\varphi\bp{\theta_{\ZZ}\np{z},y}}}
            \eqfinv
            \label{hyp_item_3}
            \\
            &\phantom{aaaaaaaaaaaaaaaaaaaaaaaaa}\forall \np{x,y} \in \XX \times \YY
            \eqfinp
            \nonumber
    \end{align}
\end{mydef}

Here is our result that provides nested formula for Fenchel-Moreau conjugates.
\begin{prop}\label{prop_rect_to_nested_generic}
    Let $\XX$, $\YY$, $\ZZ$ and $\YY'$ be four sets
    and $g: \YY \to \barRR$ be a numerical function.
    Let $\varphi: \YY' \times \YY \to \barRR$ be 
    $\np{\theta_{\XX \times \ZZ}, \theta_{\XX}, \theta_{\ZZ}}$-decomposable 
    as in Definition~\ref{def_rectangle_abstract}.
    
    Let us define the coupling $\Phi: \XX \times \np{\YY \times \ZZ} \to \barRR$ by
    \begin{equation}
        \Phi\bp{x, \np{y, z}}
        =
        \varphi\bp{\theta_{\XX \times \ZZ}\np{x,z},y}
        \eqsepv
        \forall \np{x,y,z} \in \XX \times \YY \times \ZZ
        \eqfinv
        \label{hyp_item_2}
    \end{equation}
    and the function $G: \YY \times \ZZ \to \barRR$ by
    \begin{align}
        G\np{y,z}
        =
        g\np{y} \dotplus \varphi\bp{\theta_{\ZZ}\np{z},y}
        \eqsepv \forall \np{y,z} \in \YY \times \ZZ
        \eqfinp
        \label{hyp_item_2_bis}
    \end{align}
    Then, we have the following Nested Formula between Fenchel-Moreau conjugates: 
    \begin{equation}\label{eq_NF_rect}
        G^{\Phi} = g^{\varphi} \circ \theta_{\XX}
        \eqfinp
    \end{equation}
\end{prop}

\begin{proof}
    We refer the reader to Appendix~\ref{Appendix} for the details of the proof.
\end{proof}

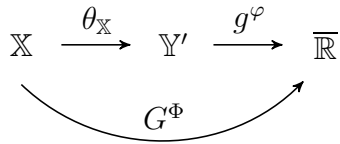
\begin{figure}[ht]
    \centering
    \begin{tikzpicture}[->,>=stealth',shorten >=1pt,auto,node distance=2cm, semithick]
        \tikzstyle{every state}=[draw=none,text=black]

        \node[state]    (A)                                 {$\XX$};
        \node[state]    (B)  [right of=A]                   {$\YY'$};
        \node[state]    (C)  [right of=B]     {$\overline{\RR}$};

        \path   
        (A)     edge  [above]   node {$\theta_{\XX}$}  (B)
        (B)     edge  [above]   node {$g^{\varphi}$}  (C)
        (A.270)     edge  [bend right=45]   node {$G^{\Phi}$}  (C.240);

    \end{tikzpicture}
    \caption{Representation of the Nested Formula~\eqref{eq_NF_rect}}
\end{figure}

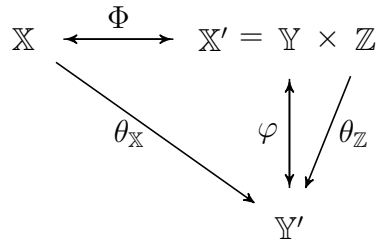
\begin{figure}[ht]
    \centering
    \begin{tikzpicture}[->,>=stealth',shorten >=1pt,auto,node distance=0.5cm, semithick]
        \tikzstyle{every state}=[draw=none,text=black]

        \node[state]    (A)                                 {$\XX$};
        \node[state]    (B)  [right of=A, xshift=2.0cm]     {$\XX'$};

        \node[state]    (C)  [right of=B]                   {$=$};
        \node[state]    (D)  [right of=C]                   {$\YY$};
        \node[state]    (E)  [right of=D]                   {$\times$};
        \node[state]    (F)  [right of=E]                   {$\ZZ$};
        \node[state]    (G)  [below of=D, yshift=-2.0cm]     {$\YY'$};

        \path   
                (A)     edge  [above]   node {$\Phi$}  (B)
                (B)     edge  [above]   node {}  (A)
                (D)     edge  [left]  node {$\varphi$}  (G)
                (G)     edge  [right]  node {}  (D)
                (A)     edge  [left]  node {$\theta_{\XX}$}   (G)
                (F)     edge  [right]  node {$\theta_{\ZZ}$}   (G);

    \end{tikzpicture}
    \caption{Representation of links between the mappings of Definition~\ref{def_rectangle_abstract}}
\end{figure}

\section{Conclusion}
\label{section_conclusion_TC_NF}

Time Consistency is a notion discussed in 
economics (dynamic optimization, bargaining) and mathematics (dynamical risk
measures, multi-stage stochastic optimization). 
We have gone through a selection of papers that are representative of the different 
fields; we have tried to separate the common core elements related to Time Consistency 
from the additional assumptions that make the specific contribution of each
author.
We have presented a framework of Weak Time Consistency which allows us to prove an equivalence
with a Nested Formula, under minimal assumptions. 
By formulating the core skeleton axioms, we hope to have shed light on the
notion of Time Consistency, often melted with other notions in the literature.
We believe that this makes the notion more transparent and we showed that 
it opens the way for possible extensions.
Indeed, we have established in Proposition~\ref{prop_sum_acceptance_set} 
a nice relation between acceptance sets of 
the original mapping $\rho$ and the ``conditional'' mapping $F$. 
In Proposition~\ref{prop_rect_to_nested_generic}, 
we have put to light an intriguing relation that certainly needs further 
investigation.

\bigskip
\textbf{Acknowledgements.}
The authors want to thank Université Paris-Est and Labex Bézout for the financial support. 
The first author particularly thanks them for the funding of his PhD program.

\appendix
\section{Appendix}
\label{Appendix}

We provide here the proofs of two Propositions 
of~\sect{Two_classes_of_time_consistent_mappings}.

\subsection{Proof of Proposition~\ref{prop_sum_acceptance_set}}

\begin{proof}
    The proof goes in three steps as follows:
    \begin{enumerate}
        \item\label{item1_proof_FandS} first, we show that 
        $\tailelement \in \cA_{\factor} \oplus \cA_{\mesu_{\mid \secondvalueset}} 
        \Leftrightarrow \factor\np{\tailelement} \in \cA_{\mesu_{\mid \secondvalueset}}
        \eqsepv \forall \tailelement \in \TailSet$,
        \item then we use the previous assertion to prove the two following statements:
        \begin{subequations}\label{eq_proof_FandS}
        \begin{align}
            \label{eq_proof_acceptset_1}
            \cA_{\mesu} \subset \cA_{\factor} \oplus \cA_{\mesu_{\mid \secondvalueset}}
            \Leftrightarrow
            \mesu\bp{\factor\np{\tailelement}}
            \leq
            \mesu\np{\tailelement}
            \eqsepv
            \forall \tailelement \in \TailSet
            \eqfinv
            \\
            \label{eq_proof_acceptset_2}
            \cA_{\mesu} \supset \cA_{\factor} \oplus \cA_{\mesu_{\mid \secondvalueset}}
            \Leftrightarrow
            \mesu\bp{\factor\np{\tailelement}}
            \geq
            \mesu\np{\tailelement}
            \eqsepv
            \forall \tailelement \in \TailSet
            \eqfinv
        \end{align}
        \end{subequations}
        \item finally, we bring all elements together to conclude.
    \end{enumerate}
    
    We now detail each step.
    \begin{enumerate}
       \item We prove the implication 
       \(
        \tailelement \in \cA_{\factor} \oplus \cA_{\mesu_{\mid \secondvalueset}} 
        \Rightarrow 
        \factor\np{\tailelement} \in \cA_{\mesu_{\mid \secondvalueset}}
        \)
        and the reverse statement 
        \(
        \factor\np{\tailelement} \in \cA_{\mesu_{\mid \secondvalueset}}
        \Rightarrow
        \tailelement \in \cA_{\factor} \oplus \cA_{\mesu_{\mid \secondvalueset}}
        \)
        successively.
       \begin{itemize}
       \item Let $\tailelement \in \cA_{\factor} \oplus \cA_{\mesu_{\mid \secondvalueset}}$ be given. 
       By definition, $\tailelement$ can be decomposed as
            $\tailelement = \tailelement_{\factor} \oplus \tailelement_{\mesu}$ with
            $\tailelement_{\factor} \in \cA_{\factor}$ and $\tailelement_{\mesu} \in \cA_{\mesu_{\mid \secondvalueset}}$.
            We successively obtain 
            \begin{align}
              \factor\np{\tailelement} 
              &= 
              \factor\np{\tailelement_{\factor}} \oplus \tailelement_{\mesu}
              \eqfinv
              \tag{as $\tailelement_{\mesu} \in \secondvalueset$ and $\factor$ is $\np{\TailSet,\secondvalueset}$-translation invariant}
              \\
              &\leq
              \tailelement_{\mesu}
              \eqfinv
              \tag{as $\tailelement_{\factor} \in \cA_{\factor} = \na{\tailelement \in \TailSet \mid \factor\np{\tailelement} \leq 0}$}
            \end{align}
            which leads to
            \begin{align}
                \mesu\bp{\factor\np{\tailelement}} 
                &\leq \mesu\np{\tailelement_{\mesu}}
                \eqfinv
                \tag{by monotonicity of $\mesu$}
                \\
                &\leq 0
                \eqfinv
                \tag{by definition of $\tailelement_{\mesu} \in \cA_{\mesu_{\mid \secondvalueset}}$}
            \end{align}
            and hence, $\factor\np{\tailelement} \in \cA_{\mesu_{\mid \secondvalueset}}$.
            \item We now assume that \( \factor\np{\tailelement} \in \cA_{\mesu_{\mid \secondvalueset}} \)
            and recall that for all $\tailelement \in \TailSet$, $\factor\bp{\tailelement \ominus \factor\np{\tailelement}} 
            = \factor\np{\tailelement} \ominus \factor\np{\tailelement} = 0$ by $\np{\TailSet-\secondvalueset}$-translation invariance
            of the mapping $\factor$.
            The converse implication follows immediately from the decomposition
            $\tailelement = \tailelement \ominus \factor\np{\tailelement} \oplus \factor\np{\tailelement}$ 
            since $\factor\np{\tailelement} \in \cA_{\mesu_{\mid \secondvalueset}}$
            by assumption and $\tailelement \ominus \factor\np{\tailelement} \in \cA_{\factor}$.
       \end{itemize}

       \item We prove statements~\eqref{eq_proof_acceptset_1} and~\eqref{eq_proof_acceptset_2} successively.
        \begin{itemize}
            \item First, we focus on equation~\eqref{eq_proof_acceptset_1}:
            \begin{align}
                \cA_{\mesu} \subset \cA_{\factor} \oplus \cA_{\mesu_{\mid \secondvalueset}}
                \Leftrightarrow
                \mesu\bp{\factor\np{\tailelement}}
                \leq
                \mesu\np{\tailelement}
                \eqsepv
                \forall \tailelement \in \TailSet
                \eqfinv
            \end{align}
            
            We suppose that left hand side of this equation is satisfied, i.e.
            $\cA_{\mesu} \subset \cA_{\factor} \oplus \cA_{\mesu_{\mid \secondvalueset}}$,
            and we show that it implies the right hand side of the equation.
            For that purpose, we fix $\tailelement \in \TailSet$.
            We recall that $\mesu\np{\tailelement} \in \valueset \subset \secondvalueset$ 
            by definition of the mapping $\mesu : \TailSet \to \valueset$ and assumption $\np{\valueset, \oplus} \subset \np{\secondvalueset, \oplus}$
            We have that $\factor\np{\tailelement} \ominus \mesu\np{\tailelement} = \factor\bp{\tailelement \ominus \mesu\np{\tailelement}}$ 
            by $\np{\TailSet, \secondvalueset}$-translation invariance of the mapping$\factor$.
            As $\tailelement \ominus \mesu\np{\tailelement} \in \cA_{\mesu} \subset \cA_{_{\factor}} \oplus \cA_{\mesu_{\mid \secondvalueset}}$ 
            we get by item~\ref{item1_proof_FandS} just above 
            that $\factor\bp{\tailelement \ominus \mesu\np{\tailelement}} \in \cA_{\mesu_{\mid \secondvalueset}}$ and then 
            $\factor\np{\tailelement} \ominus \mesu\np{\tailelement} \in \cA_{\mesu_{\mid \secondvalueset}}$.
            This implies that
            \begin{equation}
                \mesu\bp{\factor\np{\tailelement}} \ominus \mesu\np{\tailelement}
                =
                \mesu\bp{\factor\np{\tailelement} \ominus \mesu\np{\tailelement}}
                \leq
                0
                \eqfinp
              \end{equation}

            Assume now that $\mesu\bp{\factor\np{\tailelement}} \leq \mesu\np{\tailelement}$ 
            for all $\tailelement \in \TailSet$
            and let $\tilde{\tailelement} \in \cA_{\mesu}$.
            Then by definition~\eqref{eq_def_acceptance_set} 
            of an acceptance set, we got that $\mesu\np{\tilde{\tailelement}} \leq 0$. 
            It follows that $\mesu\bp{\factor\np{\tilde{\tailelement}}} \leq 0$
            and that $\factor\np{\tilde{\tailelement}} \in \cA_{\mesu_{\mid \secondvalueset}}$
            and so, by item~\ref{item1_proof_FandS} just above,
            that $\tilde{\tailelement} \in \cA_{\factor} \oplus \cA_{\mesu_{\mid \secondvalueset}}$.

            \item Second, we focus on Equation~\eqref{eq_proof_acceptset_2}
            \begin{align}
                \cA_{\mesu} \supset \cA_{\factor} \oplus \cA_{\mesu_{\mid \secondvalueset}}
                \Leftrightarrow
                \mesu\bp{\factor\np{\tailelement}}
                \geq
                \mesu\np{\tailelement}
                \eqsepv
                \forall \tailelement \in \TailSet
                \eqfinv
            \end{align}
            We assume $\cA_{\mesu} \supset \cA_{\factor} \oplus \cA_{\mesu_{\mid \secondvalueset}}$. 
            Let us fix $\tailelement \in \TailSet$.
            Then, by adding and removing the term $\factor\np{\tailelement}$ we get
            \begin{align}
                \tailelement \ominus \mesu\bp{\factor\np{\tailelement}}
                &=
                \underbrace{\tailelement \ominus \factor\np{\tailelement}}_{\in \cA_{\factor}}
                \oplus 
                \underbrace{\factor\np{\tailelement} \ominus \mesu\bp{\factor\np{\tailelement}}}_{\in \cA_{\mesu \mid \factor}}
                \in \cA_{\factor} \oplus \cA_{\mesu_{\mid \secondvalueset}}
                \eqfinp
            \end{align}
            It follows by left hand side of~\eqref{eq_proof_acceptset_2} 
            that $\tailelement \ominus \mesu\bp{\factor\np{\tailelement}}$ belongs to $\cA_{\mesu}$.
            That implies, taken together with the 
            $\np{\TailSet,\valueset}$-translation invariance of 
            the mapping $\mesu: \TailSet \to \valueset$
            \begin{align}
                \mesu\np{\tailelement} \ominus \mesu\bp{\factor\np{\tailelement}}
                =
                \mesu\bp{\tailelement \ominus \mesu\bp{\factor\np{\tailelement}}}
                \leq 0
                \eqfinp
            \end{align}

            To prove the reverse implication of Equation~\eqref{eq_proof_acceptset_2}, 
            take $\tailelement \in \cA_{\factor} \oplus \cA_{\mesu_{\mid \secondvalueset}}$ and assume that 
            $\mesu\bp{\factor\np{\valueelement}}  \geq\mesu\np{\valueelement}$.
            Using step~\ref{item1_proof_FandS}, we have that $\factor\np{\tailelement} \in \cA_{\mesu_{\mid \secondvalueset}}$ and we obtain that 
            \begin{equation}
                \mesu\np{\tailelement} \leq \mesu\bp{\factor\np{\tailelement}} \leq 0
                \eqfinv
            \end{equation}
            which gives $\tailelement \in \cA_{\mesu}$
            by definition~\eqref{eq_def_acceptance_set} of an acceptance set.
        \end{itemize}
        
        \item We finally bring all elements together. We know from Theorem~\ref{decomposition_time_consistent}
        that the couple of mappings $\np{\aggreg_{\mesu},\factor}$ is usual time consistent
        if and only if the subaggregator $\composite^{\aggreg_{\mesu}, \factor}$ 
        defined in~\eqref{definition_subaggregator}
        is a mapping increasing in its second argument and
        we have the nested formula
        $\aggreg_{\mesu}\np{\headelement, \tailelement} = \composite^{\aggreg_{\mesu}, \factor}\bp{\headelement, \factor\np{\tailelement}}$.
       
    In this case, by Definition~\ref{definition_subaggregator}, we have that
    \begin{align}
        \composite^{\aggreg_{\mesu}, \factor}\np{\headelement, \secondvalueelement}
        =
        \ba{\aggreg_{\mesu}\np{\headelement, \tailelement} \mid \factor\np{\tailelement} = \secondvalueelement}
        \eqsepv
        \forall \np{\headelement, \secondvalueelement} \in \HeadSet \times \secondvalueset
        \eqfinp
    \end{align}
    As the set-valued mapping $\composite^{\aggreg_{\mesu}, \factor}$ is a mapping, 
    choosing one element $\tailelement \in \TailSet$
    such that $\factor\np{\tailelement} = \secondvalueelement$ is sufficient to define the value
    of $\composite^{\aggreg_{\mesu}, \factor}\np{\headelement, \secondvalueelement}$.
    We notice that, for each element $\secondvalueelement \in \secondvalueset$, 
    the following statement holds true
    \begin{equation}
        \factor\np{\tailelement}
        =
        \factor\np{0} \oplus \tailelement
        \eqfinp
    \end{equation}
    By $\np{\TailSet,\secondvalueset}$-translation invariance property~\eqref{eq_IT_ominus},
    we have that $\factor\bp{\secondvalueelement \ominus \factor\np{0}} 
    = \factor\np{0} \oplus \bp{\secondvalueelement \ominus \factor\np{0}}
    =\secondvalueelement$ for all $\secondvalueelement \in \secondvalueset$.
    We deduce that
    \begin{align}\label{eq_sub_IT}
        \composite^{\aggreg_{\mesu}, \factor}\np{\headelement, \secondvalueelement}
        =
        \aggreg_{\mesu}\bp{\headelement, \secondvalueelement \ominus \factor\np{0}}
        \eqsepv
        \forall \np{\headelement, \secondvalueelement} \in \HeadSet \times \secondvalueset
        \eqfinp
    \end{align}
    Hence, the nested formula $\aggreg_{\mesu}\np{\headelement, \tailelement} 
    = 
    \composite^{\aggreg_{\mesu}, \factor} \bp{\headelement, \factor\np{\tailelement}}$
    reads
    \begin{align}
        \aggreg_{\mesu}\np{\headelement, \tailelement} 
        &= 
        \aggreg_{\mesu}\bp{\headelement, \factor\np{\tailelement} \ominus \factor\np{0}}
        \eqfinv
        \tag{by~\eqref{eq_sub_IT}}
        \\
        \mesu\np{\headelement \oplus \tailelement} 
        &= 
        \mesu\bp{\headelement \oplus \factor\np{\tailelement} \ominus \factor\np{0}}
        \eqfinv
        \tag{by~\eqref{eq_aggreg_mesu} that defines $\aggreg_{\mesu}$}
        \\
        \headelement \oplus \mesu\np{\tailelement} 
        &= 
        \headelement \oplus \mesu\bp{\factor\np{\tailelement} \ominus \factor\np{0}}
        \tag{by $\np{\TailSet, \valueset}$-translation invariance}
        \eqfinv
        \\
        \mesu\np{\tailelement} 
        &= 
        \mesu\bp{\factor\np{\tailelement} \ominus \factor\np{0}}
        \tag{by compatibility of $\oplus$ with $\leq$}
        \eqfinv
        \\
        \mesu\np{\tailelement} 
        &= 
        \mesu\bp{\factor\np{\tailelement}}
        \tag{as $\factor\np{0} = 0$.}
    \end{align}        
    The fact that $\mesu\np{\tailelement} = \mesu\bp{\factor\np{\tailelement}}$ 
    taken together with both statements 
    of Equations~\eqref{eq_proof_FandS} gives the wanted result.
        
        This ends the proof.
    \end{enumerate}
\end{proof}

\subsection{Proof of Proposition~\ref{prop_rect_to_nested_generic}}

\begin{proof}
    We have, for any $x \in \XX$, the following equalities 
    \begin{align}
        G^{\Phi}\np{x}
        &=
        \sup_{\np{y, z} \in \YY \times \ZZ}
        \Ba{
            \Phi\bp{x, \np{y, z}}
            \plusdot
            \bp{-G\np{y,z}}
        }
        \eqfinv
        \\
        \intertext{by Equation~\eqref{eq:upper-Fenchel-Moreau_conjugate} that expresses the $\Phi$-conjugate of $G$,}
        &= 
        \sup_{\np{y, z} \in \YY \times \ZZ}
        \bga{
            \varphi\bp{\theta_{\XX \times \ZZ}\np{x,z},y}
            \plusdot
            \bp{-g\np{y}} 
            \plusdot 
            \Bp{-\varphi\bp{\theta_{\ZZ}\np{z},y}}
        }
        \eqfinv
        \\
        \intertext{by Equations~\eqref{hyp_item_2} and~\eqref{hyp_item_2_bis} that express particular forms of $\Phi$ and $G$,
            and by the joint property~\eqref{eq:lower_upper_addition_minus} of Moreau's additions,}
        &= 
        \sup_{y\in \YY}
        \bga{
            -g\np{y} 
            \plusdot
            \sup_{z \in \ZZ}
            \Ba{
                \varphi\bp{\theta_{\XX \times \ZZ}\np{x,z},y}
                \plusdot
                \bp{-\varphi\bp{\theta_{\ZZ}\np{z},y}}
            }
        }
        \eqfinv
        \\
        \intertext{by property~\eqref{eq:lower_addition_sup} of Moreau's additions,}
        &= 
        \sup_{y\in \YY}
        \Ba{
            -g\np{y} 
            \plusdot
            \varphi\bp{\theta_{\XX}\np{x},y}
        }
        \eqfinv
        \\
        \intertext{by Equation~\eqref{hyp_item_3} that expresses the supremum,}
        &=
        g^{\varphi}\bp{\theta_{\XX}\np{x}}
        \eqfinv
    \end{align}
    by Definition~\ref{def_FM_conjugate} of a Fenchel-Moreau conjugate.
    This ends the proof.
\end{proof}

\newcommand{\noopsort}[1]{} \ifx\undefined\allcaps\def\allcaps#1{#1}\fi

\end{document}